\documentclass[pdftex,envcountsame,runningheads]{llncs}

\usepackage%
{hyperref}

\usepackage{orcidlink}
\usepackage{anyfontsize}
\usepackage{amsmath,amsfonts,latexsym}%
\usepackage{clrscode3e}
\usepackage{mathpartir}
\usepackage{multirow}

\usepackage[bbgreekl]{mathbbol}
\usepackage{amssymb} %
\DeclareSymbolFontAlphabet{\mathbb}{AMSb}
\DeclareSymbolFontAlphabet{\mathbbl}{bbold}

\usepackage[]{todonotes} %
\usepackage[inline]{enumitem}		%

\usepackage{xifthen}%

\usepackage{tikz}
\usetikzlibrary{calc,arrows} %
\usetikzlibrary{shapes.multipart}

\usepackage{amsfonts}
\usepackage{amssymb}
\usepackage{graphicx}
\usepackage[T1]{fontenc} 							%
\usepackage[utf8]{inputenc} 							%

\usepackage{mdframed}
\usepackage{comment}
\definecolor{shadecolor}{gray}{0.9}

{\begin{mdframed}[backgroundcolor=red!60] \noindent\textbf{Note.\ifempty{#1}{}{~(#1)}~}}%
{\end{mdframed}}

\providecommand*{\ifempty}[3]{\ifthenelse{\isempty{#1}}{#2}{#3}}

\newcommand{\ie}{\textit{i.e.}}
\newcommand{\eg}{\textit{e.g.}}
\newcommand{\cf}{\textit{cf.}}

\usepackage{mathrsfs} %
\DeclareMathOperator{\flows}{\mathcal{F}}%
\DeclareMathOperator{\cc}{\mathcal{N}}%
\DeclareMathOperator{\hc}{\hat{\cc}}
\DeclareMathOperator{\hflows}{\hat{\flows}}

\newcommand{\V}{\mathbb{V}}
\newcommand{\D}{\mathbb{D}}
\newcommand{\A}{\mathbb{A}}
\newcommand{\T}{\mathbb{T}}
\newcommand{\E}{\mathcal{E}}
\newcommand{\ora}[1][]{\ifempty{#1}{O}{O^{#1}}}		%
\newcommand{\rest}{T}		%

\newcommand{\e}{\varepsilon}

\DeclareMathOperator{\var}{Var}
\DeclareMathOperator{\args}{Arg}

\newcommand{\lift}{L}

\newcommand{\argUnion}[1]{\overline{#1}}

\title{\texorpdfstring{Computing Fixed Points \\ using Dependency Oracles}{Computing Fixed Points using Dependency Oracles}} %

\author{Giorgio Bacci%
    \orcidlink{0000-0003-4004-6049}
	\and Giovanni Bacci%
    \orcidlink{0000-0001-8529-0681}
	\and Kim G. Larsen%
    \orcidlink{0000-0002-5953-3384}
	\and Daniele Toller%
    \orcidlink{0000-0003-3007-576X}
}

 \authorrunning{G. Bacci, G. Bacci, K.\,G. Larsen, and D. Toller} 

\institute{Department of Computer Science, Aalborg University, Denmark \\ \texttt{\{grbacci,giovbacci,klg,danieleto\}@cs.aau.dk }}

\begin{document}

\maketitle

\begin{abstract}
We present global and local algorithms for solving systems of equations over Noetherian posets with a bottom element, a general setting underlying many verification problems. 
Our algorithms compute the solution of a selected variable by restricting exploration to those parts of the system required to determine its value.
We achieve this by computing variable dependencies by means of \emph{dependency oracles}. Oracles guide the exploration of the system and provide sound termination criteria for local fixed-point computation. A key advantage of our approach is its flexibility: oracles can be customized, composed, or over-approximated, offering a principled way to trade precision for performance without compromising correctness.

We evaluate our solution against existing algorithms from the literature and show that our prototype implementation is competitive and often outperforms specialized solutions, while remaining simple and adaptable across diverse application domains.
\end{abstract}

\section{Introduction}

Fixed point computation underlies a wide range of foundational techniques in computer science, from program semantics to verification. In static analysis, for example, control-flow invariants correspond to fixed points of systems of equations over abstract domains; in model checking, verifying temporal properties reduces to computing fixed points (greatest for safety, least for liveness).

The simplest, yet most used~\cite{Apt1988,KamUllman1977,PaigeTarjan1987,Puterman1994}, method for computing fixed points is Kleene iteration, named after Kleene’s celebrated theorem. 
Given a monotone function $F$ over a poset with a least element $\bot$, the least fixed point 
$\mu F$ can be computed by iteratively applying $F$, starting from $\bot$. 
If the poset is Noetherian, 
the chain $\bot \sqsubseteq F(\bot) \sqsubseteq F ( F(\bot) ) \sqsubseteq \dots$ stabilizes after finitely many steps, guaranteeing termination. 

In practice, most applications can be expressed as solving a system of equations, where an application of $F$ updates all variables simultaneously. However, in many cases one is interested in the solution of a specific variable $\tilde x$ (target), which allows in certain situations to explore only a subset of the equations, making it unnecessary to update every variable at each step. A more efficient approach performs single updates via a function $F_x$ that updates the variable $x$ and leaves all others unchanged. To make this effective, updates must follow precise selection criteria that identify the variables relevant to $\tilde x$ and prune the search space, much like the locality principles of local model-checking algorithms like~\cite{And94,ELS22,LiuS98,VL94}. 

\begin{example} \label{ex:locality}
    Consider the following system of equations on the Boolean lattice $\{\mathbf{ff},\mathbf{tt}\}$, where $\wedge$ and $\vee$ are ordinary logical conjunction and disjunction.
    \begin{align*}
        x = y \vee z \,, &&
        y = y \wedge z \,, &&
        z = \mathbf{tt} \,, &&
        w = \mathbf{tt} \vee w \,.
    \end{align*}
    Assume $x$ is the target. The updates that would increase the initial assignment $\bot = (x\colon \mathbf{ff}; y\colon \mathbf{ff}; z\colon \mathbf{ff}; w\colon \mathbf{ff})$ are for the variables $z$ and $w$. However, an easy inspection of the equations reveals that $w$ has no influence on $x$, so should be avoided and the assignment updated to $F_z(\bot) = (x\colon \mathbf{ff}; y\colon \mathbf{ff}; z\colon \mathbf{tt}; w\colon \mathbf{ff})$. Now $z$ is solved. By exploiting the laws of Boolean algebras (namely, $\mathbf{tt} = y \vee \mathbf{tt}$), one discovers that $y$ has no influence on $x$.
    Now, the only useful update is on $x$, which gives the assignment $F_x \big( F_z(\bot) \big) = (x\colon \mathbf{tt}; y\colon \mathbf{ff}; z\colon \mathbf{tt}; w\colon \mathbf{ff})$, solving $x$. A smart algorithm should terminate without any further updates of the variables.
    
\end{example}

In this paper, we propose \emph{global} and \emph{local} algorithms for computing the solution of a target variable $\tilde x$, that exploit dependency analyses on the variables of the system of equations.
These analyses serve a dual purpose: they guide the selection of the variables which are relevant to solve $\tilde x$, and they provide a precise criterion for determining when $\tilde x$ has reached its least solution.

Our main theoretical contribution is the design of a general framework for analyzing dependencies, both semantically and syntactically, on a system of equations interpreted over a Noetherian poset with a bottom element.

We identify two key relations, \emph{flow} and \emph{now}, which are used to describe how updates to one variable may influence others within the system. Intuitively, if $x$ is related to $y$, then updating $x$ can affect the future value of $y$.
These two relations support different analysis strategies.
The now relation characterizes variable dependencies by focusing on immediate influence: updating $x$ now has a direct effect on $y$. The flow relation focuses on eventual influence: a future update of $x$ has an impact on the value of $y$.
Our interest in these relations is motivated by the fact that they can be used to determine whether a variable is solved independently of the others. This property is used to provide sound termination criteria for our algorithms. 
The termination condition for the global algorithm is derived from the now relation; the one for the local algorithm from the flow relation. 
The difference is that for the former, we assume the equations are
known in advance, while in the latter, equations are discovered incrementally during computation, and only partial information is available.

Another key contribution is the introduction of \emph{dependency oracles}: functions that refine dependency analysis combining information from the current assignment and the knowledge from the system of equations. An oracle is sound if it preserves the correctness of the analysis, \ie, if it does not remove necessary dependencies. The use of oracles allows us to improve performance by avoiding the overhead of computing exact dependency analysis, while preserving correctness.

Sound oracles are closed under composition, intersection, and union, allowing for more precise analyses to be systematically derived.  
Oracles combine flexibility and simplicity as they can be defined at varying levels of precision, reused across domains, or specialized to exploit domain-specific properties.

Our global and local algorithms are implemented in a prototype tool written in Java, which we use to evaluate performance against
competing algorithms, namely, %
WKTool~\cite{JensenLSO16}, CAAL~\cite{AndersenAEHLOSW15}, and ADG~\cite{ELS22}. 
Despite its prototypical nature, our approach achieves substantial speedups across most benchmarks, reaching up to 20 times faster performance.

\paragraph{Related Work.}
Several global and local algorithms have been proposed for fixed-point computation~\cite{And94,AC88,CS93,ELS22,LiuS98,VL94}. The works most closely related to ours are~\cite{LiuS98} and~\cite{ELS22}, which rely on \emph{dependency graphs}: data structures encoding both equations and direct variable dependencies.

The dependency graphs of~\cite{LiuS98} modeled Boolean equation systems following a rigid disjunctive-normal-form structure. Later extensions to more expressive domains~\cite{EnevoldsenJLMS20,JensenLSO16,JensenKLNS19,JensenLS16,MariegaardL17} required ad hoc adaptations. Abstract Dependency Graphs (ADGs)~\cite{EnevoldsenLS19} unified these approaches by supporting assignments over Noetherian posets with customizable edge types and user-defined \emph{ignore functions}.
Our algorithms operate at the same level of generality as ADGs, and our oracles generalize their ignore functions: every sound ignore function is a sound oracle in our sense. Moreover, oracles support richer analyses while providing a systematic framework for establishing soundness.

\paragraph{Outline.}%
Section~\ref{sec:prelim} introduces the preliminaries; Section~\ref{sec:Kleene:Iteration:with:Dependency Analysis} lays out the mathematical foundation behind our algorithms, which are presented in Section~\ref{sec:algorithms}. 
Section~\ref{sec:oracles} is dedicated to oracles and local oracles, where we demonstrate their properties. %
In Section~\ref{sec:extensions}, we focus on oracles designed for systems of equations whose right-hand sides are syntactically expressed as terms over a signature of operations.
Section~\ref{Experimental Results} outlines the experimental results of our Java prototype implementations, and Section~\ref{sec:conclusions} concludes the paper. Full proofs of the results and supplementary material are in the appendix.

\section{Preliminaries and Notation}\label{sec:prelim}
A poset $(P,\sqsubseteq)$ is a set $P$ endowed with a partial order relation $\sqsubseteq$ on it (\ie, reflexive, transitive, and anti-symmetric).
A poset is \emph{Noetherian} if every ascending chain
$p_1 \sqsubseteq p_2 \sqsubseteq p_3 \sqsubseteq \ldots$
eventually stabilizes: there is $n \in \mathbb{N}$ such that $p_n = p_{j}$, for every $j\geq n$. 

A function $f \colon P \to P$ is \emph{monotone} if whenever $p \sqsubseteq p'$, then $f(p) \sqsubseteq f(p')$. By Kleene's fixed point theorem~\cite{Kleene1952-KLEITM}, a monotone function on a Noetherian poset with bottom element $\bot$ has a least fixed point $\mu f$ arising as the stabilising element of the chain $\bot \sqsubseteq f(\bot) \sqsubseteq f(f(\bot)) \sqsubseteq \dots$.
 
We will often work with sets of functions $P^X = \{ f \mid f \colon X \to P \}$ from a set $X$ to a poset $P$, with the point-wise order $f \sqsubseteq f'$ iff $f(x) \sqsubseteq f'(x)$, for all $x \in \V$. The bottom element is the function that assigns the bottom element in $P$ to each $x \in X$. When $X$ is finite, $P^X$ is Noetherian if and only if $P$ is.

\begin{definition}[System of equations] \label{def:sysEq}
Let $\D$ be a Noetherian poset with bottom element.
A \emph{system of equations} $\E$ over $\D$ consists of a finite set $\V$ of variables with an associated equation $x = f_x$, for each $x \in \V$, such that $f_x \colon \D^\V \to \D$ (the \emph{right-hand side}) is monotone.
\end{definition}

An \emph{assignment} for $\E$ is a function $A \colon \V \to \D$ mapping variables to elements in $\D$. An assignment $A$ is a \emph{solution} to $\E$ if $A(x) = f_x(A)$, for all $x \in \V$.

A system $\E$ induces a function $F \colon \D^\V \to \D^\V$ (\emph{update function}) defined as $F(A)(x) = f_x(A)$, which updates an assignment $A \in \D^\V$ by evaluating the right-hand side of each equation. As $f_x$ is monotone for every $x \in \V$, so is $F$. 

Observe that the solutions to $\E$ are exactly the fixed points of $F$. In this paper we are interested in the least solution, $\mu F$, which, under our working assumptions, always exists by Kleene fixed point theorem.
In the rest of the paper, unless stated otherwise, the components of $\E$ are denoted as in \autoref{def:sysEq}.
For instance, $\V$ denotes the set of variables and $f_x$ is the right-hand side of the equation for the variable $x$.

Given $x \in \V$ and a system of equations $\E$, let $F_x \colon \D^\V \to \D^\V$ denote the \emph{partial update} operator associated with $x$. Formally,  for $A \in \D^\V$ and $y \in \V$ 
\begin{equation*}
F_x(A)(y) =
\begin{cases}
f_x(A)  & \text{if } y = x,\\
A(y)& \text{otherwise}.
\end{cases}
\end{equation*}
As $f_x$ is monotone, so is $F_x$. Moreover, %
$F_x(A) (x) = f_x(A) = F(A)(x)$, thus $A = F(A)$ if and only if $A = F_x (A)$ for every $x\in \V$.

We will often consider sequences of updates, so we conveniently extend our notation to strings of variables $\pi = x_1 x_2 \dots x_n \in \V^*$, and define $F_\pi \colon \D^\V \to \D^\V$ as the composition %
$F_{x_n} \cdots F_{x_2} F_{x_1}$ of the single partial updates%
\footnote{Following standard conventions, the composition operator $\circ$ will often be omitted, so for example $F_{x_1 x_2}(A) = F_{x_2} F_{x_1}(A) = F_{x_2} ( F_{x_1}(A) )$. Moreover, as the empty composition is the identity function, for the empty string $\e \in \V^*$ we have $F_\e = id$.}
\emph{performed in the order of appearance in $\pi$}.

The following is a variant of the celebrated theorem by Kleene 
that characterizes the least fixed point in terms of repeated partial updates.
\begin{theorem}[Partial Kleene] \label{th:partialKleene}
For every $\pi \in \V^*$ the following hold true:
\begin{enumerate}[label=(\alph*), ref=\alph*, topsep=0ex]
\item \label{Kleene:intro:item1}
$F_\pi(\bot) \sqsubseteq F_x F_\pi(\bot)$ for every $x \in \V$, and
\item \label{Kleene:intro:item2}
there exists $\pi' \in \V^*$ such that $F_{\pi'} F_\pi(\bot)= \mu F$.
\end{enumerate}
\end{theorem}
The original Kleene's theorem considers the chain $\bot \sqsubseteq F(\bot) \sqsubseteq F ( F(\bot) ) \sqsubseteq \dots$ of assignments.
The variant above replaces this chain with the set of assignments
\begin{equation}\label{definition:blackboard:A}%
    \A = \{ F_\pi( \bot ) \mid \pi \in \V^* \} \subseteq \D^\V
\end{equation}
obtained from $\bot$ by repeatedly applying partial updates. In \autoref{th:partialKleene}, \eqref{Kleene:intro:item1} says that $F_x$ is extensive on $\A$; and \eqref{Kleene:intro:item2} that regardless of the past choices of variable updates, there is always the possibility to reach the least fixed point $\mu F$ via an appropriate selection of future updates. Thus, $\mu F$ is the top element of $\A$.

\section{Dependency Analysis on Systems of Equations}
\label{sec:Kleene:Iteration:with:Dependency Analysis}

Given a target variable $\tilde{x}$, we can iteratively apply partial updates starting from $\bot$ until $\tilde{x}$ appears to have stabilized. However, \autoref{th:partialKleene} guarantees that this value is the correct solution only when the entire system has stabilized. 

In this section, we establish criteria for certifying that a variable is solved without waiting for the remaining variables to converge.
We do this by means of two relations, \emph{flow} and \emph{now}, which characterize variable dependencies according to whether the influence is eventual or immediate.

Hereafter, we fix a system of equations $\E = \{x=f_x \mid x\in \V \}$ on a Noetherian poset $(\D, \sqsubseteq)$, with variables in $\V$ and least solution $\mu F$. The set $\A$ of assignments is defined as in equation \eqref{definition:blackboard:A}.

\begin{definition}[Flow and Now relations]\label{def:flow}
For an arbitrary $A \in \A$, we define two relations:  
$\flows_A \subseteq \V \times \V$ (\emph{flow relation}) and  $\cc_A \subseteq \V \times \V$ (\emph{now relation}), relating the variables of $\E$ according to how they influence each other 
relative to a given assignment $A$:
\begin{align*}
    \flows_A &= 
    \{ (x,y) \mid \exists \pi, \pi'\in \V^* \text{ such that }
    F_{\pi\pi'}(A)(y) \neq F_{\pi x \pi'}(A) (y) \} \,,
    \\
    \cc_A &= 
    \{ (x,y) \mid \exists \pi \in \V^* \text{ such that }
    F_{\pi}(A)(y) \neq F_{x\pi}(A)(y) \} \,.
\end{align*}
\end{definition} 
Intuitively, $(x,y) \in \flows_A$ if updating $x$ at some later stage can affect the future value of $y$, whereas $(x,y) \in \cc_A$ reflects the influence of an immediate update of $x$ on the future of $y$. 
Thus, the difference is temporal: $\flows_A$ captures complex long-term influence, while $\cc_A$ only captures direct influence. 
\begin{proposition}\label{flows:anti-monotonic}
For every $A \in \A$, it holds that that $\flows_A = \bigcup_{\pi \in \V^*} \cc_{F_\pi(A)}$ (and so $\cc_A \subseteq \flows_A$). Moreover, $\flows_{F_{x}(A)} \subseteq \flows_{A}$ for every $x \in \V$.
\end{proposition}

\begin{example} \label{ex:flow}\label{ex:now}
    Consider the system of equations from \autoref{ex:locality}. We provide the flow and now relations for the assignments $\bot$, $F_z(\bot)$ and $F_{zx}(\bot)$:
    \begin{gather*}
        \cc_\bot = \{(z,x), (z,z), (w,w) \} \subseteq \flows_\bot = \{(x,x), (z,x), (z,z), (w,w) \} \,,
        \\
        \begin{gathered}
            \cc_{F_z(\bot)} = \{ (x,x), (w,w) \} = \flows_{F_z(\bot)},
            \qquad  \cc_{F_{zx}(\bot)} = \{ (w,w) \} = \flows_{F_{zx}(\bot)} \,. 
        \end{gathered}
    \end{gather*}
    Note that $(x,x) \notin \cc_\bot$, as $\bot = F_x(\bot)$; moreover, $\cc_{F_{z}(\bot)} \not\subseteq \cc_{\bot}$.
    In contrast, $(x,x) \in \flows_\bot$, as the value of $x$ may change in the future, for instance after the update of $z$ has taken place: indeed, $\bot \neq F_{zx}(\bot)$.
\end{example}

Our interest in these two relations is motivated by the following result, which yields a number of equivalent criteria to establish when a variable $z$ is locally solved at a given assignment $A$ (\ie, $A(z) = \mu F(z)$).

\begin{theorem} \label{th:flowsLocalFix} \label{th:nowLocalFix}
For every $A\in\A$ and $z\in \V$, the following equivalences hold
\begin{align*}
A(z) = \mu F(z) 
&\iff \forall x \in \V.\, (x,z) \notin \flows_A 
\iff \forall y \in \V.\, (z,y) \notin \flows_A \\ 
&\iff \forall x \in \V.\, (x,z) \notin \cc_A
\iff (z,z) \notin \flows_A \,.
\end{align*}
\end{theorem}
Thus, solved variables are precisely those that can neither be influenced nor influence any other variable, now or in the future.

\section{Partial Kleene Iteration with Dependency Oracles}
\label{sec:algorithms}

We present two algorithms for 
computing the least solution of a target variable based on partial Kleene iteration. The key novelty is the use of \emph{dependency oracles}, which help both to determine when the solution has been reached and to narrow the selection of the updates.

Hereafter, we assume all functions $f_x$ appearing in the system of equations are computable. This is a reasonable assumption that is needed only to prove the termination of the algorithms.

\subsection{Global algorithm}
\label{sec:DependencyKleeneAlgorithm}

The result of a dependency analysis is a relation $R \subseteq \V \times \V$ that relates variables according to how they influence each other. The analysis $R$ is correct relative to the assignment $A$ if $\cc_A \subseteq R$ (\ie, it does not exclude real dependencies).

Our approach uses functions, called oracles, that refine dependency analysis using the information from the system of equations and the current assignment.
\begin{definition}[Oracle] \label{def:oracle}
We call a function $\ora \colon \A \times \wp (\V \times \V) \to \wp (\V \times \V)$ an \emph{oracle} for $\E$. 
An oracle is sound if, 
for every $A \in \A$ and $R \subseteq \V \times \V$,
\begin{equation*}
    \text{$\cc_A \subseteq R$ implies $\cc_A \subseteq \ora(A, R)$} \,.
\end{equation*}
\end{definition}
The soundness condition simply states that every correct dependency analysis $R$, once refined by the oracle, remains correct.

\begin{example} \label{ex:EasyOracles}
The following are some simple examples of sound oracles.
\begin{itemize}%
    \item \emph{Identity oracle.} 
    The function $\mathit{i}(A , R)=R$;

    \item \emph{Trivial oracle.}
    The constant function $\mathit{triv}(A , R) = \V \times \V$;

    \item \emph{Exact oracle.}
    The function $\mathit{now}(A , R) = \cc_A$.

    \item \emph{Not-stuck oracle.}
    The function $\mathit{ns}(A,R) = \{ (x,y) \mid A(x) \neq f_x(A), y\in \V \}$.
    
    \item \emph{Reachable args oracle.} The constant function $\mathit{args}(A, R) = {\to^*}$, given by the reflexive 
    and the transitive closure of the relation $y \to x$ iff $y \in \args(f_x)$, where $\args(f)$ denotes the set of arguments %
    of a function $f \colon \D^\V \to \D$, \ie,  the smallest set $S \subseteq \V$ such that, there exists $g \colon \D^S \to \D$ with $f(A) = g(A|_S)$ for all $A \in \A$.
\end{itemize}

\begin{remark}
While the exact oracle, $\mathit{now}$, could in principle be used, computing it is impractical, as it is as costly as solving the entire system of equations. Using a generic sound oracle $\ora$ is a practical alternative: we replace an exact analysis with an over-approximation. %
In Sections~\ref{sec:oracles}--\ref{sec:extensions} we present more elaborate examples of sound oracles, and a detailed discussion on how to craft them.
\end{remark}
\end{example}

The algorithm $\proc{GlobalK}$ (\autoref{fig:Algorithms}) takes as input a system of equations $\E$, a target variable $\tilde{x}$, and a sound oracle $\ora$. It uses three variables: an assignment $A$,  a dependency analysis $R$, and the set $\id{todo}$ of variables to be considered for an update. 
The analysis $R$ is updated by invoking the oracle on the current assignment $A$ and the trivially correct analysis $\V \times \V$ (lines~\ref{li:initR} and \ref{li:updateR})%
\footnote{%
We always use $\V \times \V$ instead of using the current analysis $R$ because it may happen that $\cc_{F_x(A)} \not\subseteq \ora(F_x(A),R)$ even though $\cc_A \subseteq R$ (\cf~\autoref{ex:now}).%
}.
The analysis is used to guide the selection of the variables to update (line~\ref{li:extract}) by collecting in $\id{todo}$ only the variables that may have an influence on $\tilde{x}$. Initially, $A = \bot$. At each iteration of the while-loop (lines~\ref{li:whilebegin}--\ref{li:updatetodo}), a variable is extracted from $\id{todo}$ to undergo an update. If the assignment increases, the analysis $R$ is updated accordingly (lines~\ref{li:increase}--\ref{li:updatetodo}). 
At the end of each iteration, either the assignment has increased or $\id{todo}$ has become smaller. %
Since the domain of the assignments is Noetherian, $A$ eventually stabilizes, and $\id{todo}$ will become empty. 

By the soundness of $\ora$,
the inclusion $\{ y \mid (y,\tilde{x}) \in \cc_A \} \subseteq \id{todo}$ is satisfied throughout the computation. Thus, by 
\autoref{th:nowLocalFix}, when $\id{todo}$ becomes empty we have that $\tilde{x}$ is solved in $A$.

\begin{figure}[t]
\centering
\begin{minipage}{0.45\textwidth}
\begin{codebox}
\Procname{$\proc{GlobalK}(\E,\tilde x,\ora)$}
\li $A \gets \bot$
\li $R \gets \ora(A,\V \times \V)$ \label{li:initR}
\li $\id{todo} \gets \{ y \mid (y,\tilde x) \in R \}$ \label{li:initTodo}
\li \While $\id{todo} \neq \emptyset$ \label{li:whilebegin}
\Do
\li extract $x$ from $\id{todo}$ \label{li:extract}
\li \If $A(x) \neq f_x(A)$ \label{li:increase}
\Then   
\li $A(x) \gets f_x(A)$   \label{line:updateassign}
\li $R \gets \ora(A,\V \times \V)$ \label{li:updateR}
\li $\id{todo} \gets \{ y \mid (y,\tilde x) \in R \}$ \label{li:updatetodo} 
\End 
\End 		
\li \Return $A(\tilde x)$
\end{codebox}
\end{minipage}
\hfill
\begin{minipage}{0.5\textwidth}
\begin{codebox}
\Procname{$\proc{LocalK}(\E,\tilde{x},\ora)$}
\li $A \gets \bot$, $D = \{ \tilde{x} \}$, $V = \emptyset$ \label{li:LocalKleene:initialization}
\li $R \gets \ora(V, A,\V \times \V)$
\li $\id{todo} \gets \{ y \mid (y,\tilde x) \in R \} \cap D$ \label{li:LocalKleene:initTodo}
\li \While $\id{todo} \neq \emptyset$
\label{li:LocalKleene:while:test}
\Do
\li extract $x$ from $\id{todo}$ \label{li:LocalKleene:extract}
\li $V = V \cup \{x\}$  \label{li:LocalKleene:visited}
\li \If $A(x) \neq f_x(A)$ or $\args(f_x) \not\subseteq D$  \label{li:LocalKleene:Increase}
\Then   
\li $A(x) \gets f_x(A)$  \label{li:LocalKleene:AssUpdate}
\li $D = D \cup \args(f_x)$ \label{li:LocalKleene:DiscoveredUpdate}
\li $R \gets \ora(V, A, R)$  \label{li:LocalKleene:RelationUpdate}
\li $\id{todo} \gets \{ y \mid (y,\tilde x) \in R \} \cap D$
\label{li:LocalKleene:todoUpdate}
\End 
\End 		
\li \Return $A(\tilde x)$
\end{codebox}
\end{minipage}
    \caption{Two Kleene iteration algorithms enhanced with global (on the left) and local (on the right) dependency analysis, performed by a user-provided oracle $\ora$. 
        }
    \label{fig:Algorithms}
\end{figure}

\begin{theorem}%
\label{thm:depKleene:correct}
    If $\ora$ is a sound oracle,
    $\proc{GlobalK}(\E, \tilde{x}, \ora)$ terminates and returns $\mu F(\tilde{x})$.
\end{theorem}

\begin{example} \label{ex:global:algorithm:boolean:oracle}
The computation informally described in \autoref{ex:locality} corresponds to the formal execution of $\proc{GlobalK}$ shown below. For clarity, we explicitly report the dependency analysis $R$ used at each iteration, anticipating that it is obtained using the Boolean oracle introduced later in Section~\ref{sec:extensions}.

\medskip
{\small
\centering
\begin{tabular}{|@{\;}c@{\;}|@{\;}c@{\;}|@{\;}c@{\;}|@{\;}c@{\;}|}
\hline
Iteration & $A$ & $R$ & \id{todo} %
\\
\hline\hline
0 & $\bot$ & $\{ (z,z), (w,w) \} {\cup} \big( \{y,z,w\} \times \{x\} \big) {\cup} \big( \{x,z,w\} \times \{y\} \big)$  & $\{y,z,w\}$ \\
1 (extract $z$) & $F_z(\bot)$ & $\{ (x,x), (w,w) \} \cup \big( \{x,z,w\} \times \{y\} \big)$  & $\{x\}$ \\
2 (extract $x$) & $F_{zx}(\bot)$ & $\{ (w,w) \} \cup \big( \{x,z,w\} \times \{y\} \big)$  & $\emptyset$ \\
\hline
\end{tabular}
}
\end{example}

\subsection{Local algorithm}
\label{sec:LocalKleeneAlgorithm}

For some systems, the equations are produced as needed during execution, and the full system need never be constructed. Algorithms that follow this principle are known as \emph{local}. %
In our local algorithm, the exploration proceeds incrementally, much like traversing a graph: the equation %
is constructed only when its variable is \emph{visited}. Once visited, its equation remains available for later use (\eg, for a dependency analysis), and all new variables occurring in the equation are marked as \emph{discovered}.
To ensure locality, the next variable to update is chosen from the set of already discovered variables, rather than the entire system. This restriction adds two challenges. First, the \emph{now} relation is no longer a sound basis for dependency analyses, since variables that may influence $\tilde{x}$ according to \emph{now} might not yet have been discovered (\cf~\autoref{local:soundness:requires flows}); we address this by relying on the flow relation, which provides stronger soundness guarantees.
Second, oracles must perform their analysis using only the information available at the moment they are invoked. To support this, we extend oracles with an additional parameter that specifies which variables have been visited.

\begin{definition}[Local Oracle] \label{def:locOracle}
We call $\ora \colon \wp(\V) \times \A \times \wp (\V \times \V) \to \wp (\V \times \V)$ a \emph{local oracle} for $\E$. 
It is sound if, for every $V \subseteq \V$, $A \in \A$, and $R \subseteq \V \times \V$,
\begin{equation*}
    \flows_A\subseteq R 
    \;\text{ implies }\;
    \flows_A  \subseteq \ora(V,A, R) \,.
\end{equation*}
\end{definition}

Apart from the difference in type, $\mathit{i}$ and $\mathit{triv}$ are sound local oracles, whereas $\mathit{now}$ is not. The local variant of the exact oracle is $\mathit{flow}(V, A , R) = \flows_A$.

\smallskip

For the local exploration, $\proc{LocalK}$ (in \autoref{fig:Algorithms}) uses three sets: $\id{D}$ contains the \emph{discovered} variables; $\id{V}$ the \emph{visited} ones; and $\id{todo}$ the variables that the local oracle $\ora$ suggests to consider for the next update. Observe that the inclusion $\id{todo} \subseteq D$ is enforced throughout the entire computation.

Initially (line~\ref{li:LocalKleene:initialization}), only $\tilde x$ is known ($\id{D} = \{ \tilde{x} \}$) but its equation is not yet available (thus, $\id{V} = \emptyset$). Consequently, $\id{todo}$ is either empty or contains only $\tilde{x}$.
In the first case, the algorithm terminates immediately, as $\tilde x$ is already solved (\autoref{th:flowsLocalFix}). Otherwise, $\tilde x$ is extracted (line~\ref{li:LocalKleene:extract}) and its equation is revealed (line~\ref{li:LocalKleene:visited}). If this update increases the assignment $A$, or if the equation introduces previously undiscovered variables (line~\ref{li:LocalKleene:Increase}), $\id{todo}$ is recomputed through a new dependency analysis using the local oracle $\ora$ (line~\ref{li:LocalKleene:todoUpdate}) by refining the analysis $R$ from the previous iteration (line~\ref{li:LocalKleene:RelationUpdate}).
This process is iterated until $\id{todo}$ becomes empty (line~\ref{li:LocalKleene:while:test}).
The reuse of previous analyses (line~\ref{li:LocalKleene:RelationUpdate}) is a feature which is absent in $\proc{GlobalK}$, and it is based on the last part of~\autoref{flows:anti-monotonic}.

\begin{theorem}\label{local:version:thm:depKleene:correct}
If $\ora$ is a sound local oracle,  $\proc{LocalK}(\E, \tilde{x}, \ora)$ terminates and returns $\mu F(\tilde{x})$.
\end{theorem}
Termination follows similarly to $\proc{GlobalK}$. Establishing correctness, however, is more challenging. 
When $\proc{LocalK}$ terminates, $\id{todo}$ is empty, and this can occur in one of two ways: either because line~\ref{li:LocalKleene:initTodo} or~\ref{li:LocalKleene:todoUpdate} clears $\id{todo}$, or because repeated removals at line~\ref{li:LocalKleene:extract} eventually exhaust it.
In the first case, \autoref{th:flowsLocalFix} guarantees that $\tilde{x}$ is solved.
In the second case, $\id{todo}$ contained only variables that failed the test in line~\ref{li:LocalKleene:Increase}, meaning that updating them neither increased the assignment nor did their equations allow the discovery of new variables. Some of these variables may be unsolved (see \autoref{ex:early-termination}), so why should $\tilde{x}$ itself be solved? 

To answer this, we introduce the notion of \emph{almost self-closed set}. A set $X \subseteq \V$ is almost self-closed for $\tilde{x}$ and $A \in \A$, if  $\tilde{x} \in X$ and every $x \in X$ satisfies:
\begin{equation*}
    A(x) = f_x(A) \quad\text{and}\quad \{ y \mid (y, \tilde{x}) \in \flows_A \} \cap \args(f_x) \subseteq X \,.
\end{equation*}
where $\args(f)$ are the arguments of $f$.
Intuitively, an almost self-closed set for $\tilde{x}$ contains all the variables that possibly influence $\tilde{x}$ but none of them can increase. %

\begin{example}\label{ex:self-closed:set}
Let $x$ be the target variable of the Boolean equation system below:
\begin{align*}
    x = y \wedge z \,, &&
    y = y \vee w \,, &&
    z = \mathbf{ff} \,, &&
    w = \mathbf{tt} \,.
\end{align*}
All the almost self-closed sets for $x$ and $\bot$ are $\{x\}$, $\{x,y\}$, $\{x,z\}$ and $\{x,y,z\}$. Indeed, by \autoref{th:flowsLocalFix} $\{ v \mid (v,x) \in \flows_{\bot} \} = \emptyset$, as $\mu F(x) = \mathbf{ff}$.
\end{example}

\begin{theorem}
\label{thm:almost:self-closed}
If $\tilde{x}$ and $A$ admit an almost self-closed set, then $A(\tilde{x}) = \mu F(\tilde{x})$.
\end{theorem}

Theorem~\ref{thm:almost:self-closed} justifies the correctness of \proc{LocalK} in the second case of termination discussed earlier: when $\id{todo}$ is set for the last time in line~\ref{li:LocalKleene:todoUpdate}, it becomes an almost self-closed set for $\tilde{x}$ and $A$. Indeed, for all $x \in \id{todo}$, $A(x) = f_x(A)$ and
\begin{equation*}
    \{ y \mid (y, \tilde{x}) \in \flows_A \} \cap \args(f_x) \subseteq \{ y \mid (y, \tilde{x}) \in R\} \cap D = \id{todo}
\end{equation*} 
because $\flows_A \subseteq R$ and $\args(f_x) \subseteq D$ for all $x \in \id{todo}$. 

An advantage of almost self-closed sets is that \proc{LocalK} can terminate before all visited variables are fully solved. Example~\ref{ex:early-termination} illustrates this behavior.

\begin{example} \label{ex:early-termination}
Consider the system of equations from \autoref{ex:self-closed:set} and take $x$ as target variable. Pick a local oracle $\ora$ such that $\{ v \mid (v,x) \in\ora(V,\bot,R) \} = \{x,y,z\}$ for every $V \subseteq \V$ and $R \subseteq \V\times \V$. 
Below are shown the details of an execution of $\proc{LocalK}$:
{\small
\begin{center}
\begin{tabular}{|@{\;}c@{\;}|@{\;}c@{\;}|@{\;}c@{\;}|@{\;}c@{\;}|@{\;}c@{\;}|}
\hline
Iteration & \, $A$ \, & \id{D} (\emph{discovered}) & \id{V} (\emph{visited}) &  \id{todo} (\emph{to update}) \\
\hline\hline
0 & $\bot$ & $\{x\}$ & $\emptyset$ & $\{x\}$ \\
1 (extract $x$) & $\bot$ & $\{x, y, z\}$ & $\{x\}$ & $\{x,y,z\}$ \\
2 (extract $y$) & $\bot$ & $\{x,y,z, w \}$ & $\{x,y\}$ & $\{x,y,z\}$ \\
3 (extract $z$) & $\bot$ & $\{x,y,z, w \}$ & $\{x,y,z\}$ & $\{x,y\}$ \\
4 (extract $y$) & $\bot$ & $\{x,y,z, w \}$ & $\{x,y,z\}$ & $\{x\}$ \\
5 (extract $x$) & $\bot$ & $\{x,y,z, w \}$ & $\{x,y,z\}$ & $\emptyset$ \\
\hline
\end{tabular}
\end{center}
}
In the last three iterations of the while-loop, $x$, $y$ and $z$ are extracted from $\id{todo}$ by failing the test in line~\ref{li:LocalKleene:Increase}. From \autoref{ex:self-closed:set}, we know that previous to these extractions, $\id{todo}$ is an almost self-closed set for $x$ and $\bot$. Observe that, while the algorithm terminates correctly solving  $x$, the variable $y$ is not solved yet. 
In comparison to the local algorithm based on dependency graphs from~\cite{ELS22}, $\proc{LocalK}$ terminates earlier, because the ``early termination'' in~\cite{ELS22} still requires all the variables in $\id{todo}$ to be solved.
\end{example}

\begin{remark}\label{local:soundness:requires flows}
Next we show that the now relation fails to capture the necessary dependencies, thereby justifying the use of the flow relation for local oracles.

Suppose, for the sake of contradiction, that we execute $\proc{LocalK}$ on the system of equations from \autoref{ex:locality} with target variable $x$, using $\mathit{now}$ as oracle, which is sound in the sense of \autoref{def:oracle}.
At the start of the execution, $A = \bot$, and  $\id{D} = \{x\}$. From \autoref{ex:now}, we know that $\cc_\bot = \{(z, x), (z, z), (w, w)\}$. As a result, $\proc{LocalK}$ (line~\ref{li:LocalKleene:initTodo}) initializes $\id{todo} = \{z\} \cap \id{D} = \emptyset$. This causes the algorithm to terminate immediately, incorrectly returning $\mathbf{ff}$ as the least solution for $x$.
\end{remark}

\section{A Compositional Theory of Oracles}
\label{sec:oracles}

This section demonstrates the flexibility and ease of working with oracles and their local variants, which stem from their compositional properties.
While technically distinct, global and local oracles share similar foundations.
We will spend most of the space describing global oracles. Once these are understood, it is easy to adapt them to their local variants, which we do at the end of the section.

One wishes to obtain new oracles from existing ones by composing them via operations. The easiest operations one can think of  are union, intersection, and composition. Formally, for 
$\ora, \ora' \colon \A \times \wp(\V \times \V) \to \wp(\V \times \V)$ oracles, define
\begin{align*}
    (\ora \cup \ora')(A,R) &= \ora(A,R) \cup \ora'(A,R) \,,
    \tag{\textsc{union}} \\
    (\ora \cap \ora')(A,R) &= \ora(A,R) \cap \ora'(A,R) \,,
    \tag{\textsc{intersection}} \\
    (\ora \circ \ora')(A,R) &= \ora(A,\ora'(A,R)) \,.
    \tag{\textsc{composition}}
\end{align*}
It is clear that the operations of union and intersection above can be generalized to an arbitrary (possibly infinite) number of oracles.

An oracle $\ora$ is \emph{more accurate than} $\ora'$, written $\ora \preceq \ora'$, if $\ora(A,R) \subseteq \ora'(A,R)$ holds for all $A \in \A$ and $R \subseteq \V \times \V$.

\begin{proposition}[Closure properties] \label{prop:oracles:properties}
\label{lem:closureProp} ~
\begin{enumerate}[topsep=0.5ex]
    \item \label{upclosure}
    Sound oracles are upward closed, \ie, if $\ora \preceq \ora'$ and $\ora$ is sound, so is $\ora'$.
    \item  \label{operation-closure}
    Sound oracles are closed under union, intersection, and composition.
\end{enumerate}
\end{proposition}

The closure properties can be used to derive new sound oracles from existing ones, improve the accuracy, or even help prove their soundness. %
For example, \autoref{prop:oracles:properties} implies that every sound oracle $O$ can be made deflationary (\ie, such that $O(A,R)\subseteq R$ for all $A,R$) by intersecting it with the identity oracle: $O \cap \mathit{i}$. The resulting oracle computes $O(A,R)\cap R$.
As this operation does not add any computational cost, it is applied implicitly in all our experiments. 
Moreover, we can define the \emph{downward closure} of an oracle $\ora$, 
as $\ora^{\downarrow} = \bigcap_n \ora^n$, where $\ora^0 = \mathit{i}$ and $\ora^{n+1} = \ora \circ \ora^n$. If $\ora$ is sound, so is its downward closure.
This operation is evidently more costly and, if used, its application will be stated explicitly.

Next we present a selection of oracles, chosen to illustrate different levels of complexity and practical use. A broader collection is provided in~Appendix~\ref{sec:moreOracles}. %

\paragraph{Maximality Oracles.}%

Every poset $\D$ that is Noetherian has \emph{maximal elements}, \ie, elements $m \in \D$ such that if $m \sqsubseteq s$, then $s = m$.

Since $\mu F$ is the top element in $\A$, it follows that if a variable $z$ is assigned a maximal element, this must be its solution, $\mu F(z)$. 
According to this intuition, we define two oracles, respectively called \emph{maximality oracle} and \emph{simplified maximality oracle}, which depend only on the assignment $A \in \A$:
\begin{align*}
    &\mathit{max}(A, R) = \{ (x,x) \mid A(x) \neq f_x(A) \} \cup \{ (x,y) \mid \text{$A(x)$ and $f_y(A)$ not maximal} \} \,, 
    \\
    &\mathit{smax}(A, R) = \{ (x,y) \mid \text{$A(x)$ and $A(y)$ not maximal} \} \,.
\end{align*}
The oracle $\mathit{smax}$ discards variables that are assigned a maximal value, as updating them would have no influence on other variables. The oracle $\mathit{max}$ is similar but more sophisticated: if $y$ is not solved (\ie, $A(y) \neq f_y(A)$) but $f_y(A)$ is maximal, then $y$ is just one step away from being solved, and its value will not be influenced by any other variable apart from itself.

\begin{example}\label{example:continued}
    Consider the system of equations from \autoref{ex:locality}. For the assignment $F_z(\bot) = (x\colon \mathbf{ff}; y\colon \mathbf{ff}; z\colon \mathbf{tt}; w\colon \mathbf{ff})$ and an arbitrary $R \subseteq \V \times \V$, we have
    \begin{align*}
        \mathit{smax}(F_z(\bot), R) 
        &= \{ x,y,w\} \times \{ x,y,w \} \,,
        \\
        \mathit{max}(F_z(\bot), R) &= \{ (x,x), (x,y), (y,y), (w,y), (w,w) \} \,.
    \end{align*}
    Both oracles identify that $z$ is solved, removing it from the dependencies. The oracle $\mathit{max}$ is smarter: it identified $x$ and $w$ to be independent of other variables. 
\end{example}

The closure properties are especially useful for reducing the complexity of soundness proofs: by \autoref{prop:oracles:properties}\eqref{upclosure}, it suffices to establish soundness for a more accurate oracle. For this reason, we define the following oracles
\begin{align*}
    &\mathit{max}_l(A, R) = \{ (x,x) \mid A(x) \neq f_x(A) \} 
    \cup \{ (x,y) \mid x \neq y,  \text{ $A(x)$ not maximal} \} \,,
    \\
    &\mathit{max}_r(A, R) = \{ (x,x) \mid A(x) \neq f_x(A) \} 
    \cup \{ (x,y) \mid x \neq y,  \text{ $f_y(A)$ not maximal} \} \,.
\end{align*}

\begin{proposition} \label{lem:soundnessMaxL/R}
    $\mathit{max}_l$ and $\mathit{max}_r$ are sound, and
    $\mathit{max}_l \cap \mathit{max}_r \preceq \mathit{max} \preceq \mathit{smax}$.
\end{proposition}
We note that $\mathit{max}_l \cap \mathit{max}_r$ is strictly more accurate than $\mathit{max}$ %
(see~\autoref{ex:max'<max}), making it a natural choice when higher accuracy is desired.

\subsection{Local Oracles}\label{sec:local:oracles}

To be compatible with a local exploration, local oracles are allowed to use only the equations of the variables that have been visited. If the equation is not available, they should produce an over-approximated analysis to preserve correctness without breaking locality.

The example below shows how to obtain local oracles as straightforward variants of their non-local counterparts. Let $V$ denote the set of variables whose equation is available, and let $V^c = \V \setminus V$. We define the local variant of $\mathit{max}_r$ as
\begin{align*}
\mathit{locmax}_r(V, A, R) = {}
&\big( \V \times V^c \big) \cup 
\{ (x,y) \mid y \in V, \ x \neq y, \text{ $f_y(A)$ not maximal} \}\cup {} 
\\
&\{ (x,x) \mid x \in V, \big( A(x) \neq f_x(A) \text{ or } f_x(A) \text{ not maximal} \big) \}  \,.
\end{align*}
The first component $\V \times V^c$ reflects that we conservatively assume that every variable may influence an unvisited variable. The last two come from restricting the use of $f_y$ only to the variables $y \in V$ in the definition of $\mathit{max}_r$. Although $\{ (x,x) \mid x \in V , A(x)\neq f_x(A) \}$ seems a more natural choice than the third component, its use is unsound, as $A(x) \neq \mu F(x) \iff (x,x)\in\flows_A$ (by~\autoref{th:flowsLocalFix}).

\begin{proposition}\label{locmaxr:sound}
    $\mathit{locmax}_r$ is sound.
\end{proposition}

Oracles that do not use equations in their definitions, \eg, $\mathit{smax}$, are automatically local in the sense described above and can be used without variations.

Local oracles satisfy the same closure properties as stated in~\autoref{lem:closureProp}. %
Additional examples of local oracles are presented in~Appendix~\ref{sec:more:Local:Oracles}.

\section{Syntax-driven Oracles}
\label{sec:extensions}

In most applications, the right-hand sides of the equations are described using expressions. Although many dependency analyses can be carried out directly on the syntax of the expressions (\eg, by skipping the evaluation of sub-terms that are irrelevant to computing the least solution), the oracles we have seen so far are not designed to do so. In this section, we see how to obtain such oracles and establish general techniques to prove their soundness.

Let $\T := \T(\Sigma, \V)$ be the set of terms over a signature $\Sigma$ and variables $\V$. Denote by $\var(t) \subseteq \V$ the set of variables appearing in $t\in \T$. We assume every $n$-ary operation $\sigma \in \Sigma$ to have a monotonic interpretation $\sigma \colon \D^n \to \D$. 
The interpretation of a term $t \colon \D^\V \to \D$ is just the homomorphic extension of $A \in \D^\V$, thus is also monotonic. The generic system of equations we work with in this section has the form $\E = \{x = t_x \mid x\in \V\}$ where $t_x$ are terms.

We generalize the notion of now relation to terms as follows:
\begin{equation*}
    \hc_A =
    \{(x,t) \mid \exists \pi' \in \V^* 
        \text{ such that }
        t (F_{\pi'}(A)) \neq t (F_{x \pi'} (A)) \} \,.
\end{equation*}
Intuitively, $(x,t) \in \hc_A$ if the update of $x$ influence the future evaluation of the term $t$. Observe that $\hc_A$ coincides with $\cc_A$ on variables. The next result motivates the use of $\hc_A$ for the analysis of dependencies (compare it to \autoref{th:nowLocalFix}).
\begin{theorem}\label{improved:lemmaGR}
For $t \in \T$ and $A\in\A$, %
$
t(A) = t(\mu F) \iff 
\forall x \in \V.\, (x,t)\notin \hc_A
$.
\end{theorem}

The same considerations that lead us to the definition of oracles support the introduction of the concept of \emph{extension} and the oracles induced from it.
\begin{definition}[Extension]\label{def:property:NEW:liftings}
We call a function $\lift \colon \A \times \wp(\V \times \V) \to \wp(\V \times \T)$ an extension for $\E$. An extension
is sound if $\cc_A \subseteq R$ implies $\hc_A \subseteq \lift(A, R)$, for every $A\in\A$ and $R\subseteq \V \times \V$. 
The \emph{oracle induced by} $\lift$ is 
\begin{equation*}
\ora[\lift](A,R) = \{ (x,x) \mid A(x) \neq t_x(A) \} \cup \{(x,y) \mid x \neq y, \ (x,t_y) \in \lift(A,R) \}.
\end{equation*}
\end{definition}

An extension takes a dependency analysis $R$ on variables and produces a refined one extended on terms ---hence the name.

\begin{proposition}\label{prop:sound:oracles:from:liftings}
A sound extension $\lift$ induces a sound oracle $\ora[\lift]$.
\end{proposition}
For $\ora[\lift]$ to be sound, note that the first part in its definition, although independent of $\lift$, is necessary, since it is not difficult to see that $A(x) \neq t_x(A)$ if and only if $(x,x) \in \cc_A$ (see~\autoref{char:xxINflows:xxINnow}(\ref{char:xxINflows:xxINnow:item2})).

The accuracy relation $\lift \preceq \lift'$ on extensions is defined as expected, and entails $\ora[\lift] \preceq \ora[\lift']$. 
Extensions can be combined by means of the following operations:
\begin{gather*}
    (\lift \cap \lift')(A,R) = \lift(A,R) \cap \lift'(A,R) \,,
    \tag{\textsc{intersection}}
    \\
    \argUnion{\lift}(A,R) = \{ (x, t ) \mid (x,y) \in \lift(A,R) \text{ for some } y \in \var(t) \}.
    \tag{\textsc{Arg-Union}}
\end{gather*}

\begin{proposition}\label{lem:extension:closure:prop}
Sound extensions are closed under intersection and arg-union.
\end{proposition}
An immediate consequence of~\autoref{lem:extension:closure:prop} is that the accuracy of every extension $\lift$ can be improved without compromising soundness by using $\lift' = \lift \cap \argUnion{\lift}$.

Next we provide a paradigmatic example of extension, tailored to Boolean equation systems. Further examples are presented in~Appendix~\ref{sec:additional:extensions}.

\paragraph{Boolean Extension.}
We turn our attention to Boolean equation systems, where variables are interpreted over the Boolean lattice $\{\mathbf{ff},\mathbf{tt}\}$ and the right-hand sides are positive Boolean formulas%
\footnote{Observe that negation is not permitted. This ensures that all terms generated by the grammar above have a monotonic interpretation.} over those variables:
\begin{align*}
    t := y \mid \mathbf{tt} \mid \mathbf{ff} \mid t_1 \land t_2 \mid t_1 \lor t_2 \,.
    &&
    (y \in \V)
\end{align*}
For these systems, we introduce the extension $\mathit{bool}$, which refines dependency analyses by exploiting the algebraic properties of Boolean expressions: 
$t \wedge \mathbf{ff} = \mathbf{ff}$ and $t \vee \mathbf{tt} = \mathbf{tt}$ (\emph{absorbing elements}); $t \wedge \mathbf{tt} = t$ and $t \vee \mathbf{ff} = t$ (\emph{null elements}).

For $A\in\A$ and $R \subseteq \V \times \V$, we define $\mathit{bool}(A,R)$ as the smallest set $S \subseteq \V \times \T$ closed under the following rules of inference:
\begin{mathpar}
\frac{(x,y)\in R \quad A(y)=\mathbf{ff}}{(x,y)\in S}
\and
\frac{(x,t_i)\in S \quad (z,t_{\neg i})\in S \quad t_0(A)=\mathbf{ff}=t_1(A)}{(x,t_0 \wedge t_1)\in S}
\and
\frac{(x,t_i)\in S \quad t_i(A)=\mathbf{ff} \quad t_{\neg i}(A)=\mathbf{tt}}{(x,t_0 \wedge t_1)\in S}
\and
\frac{(x,t_i)\in S \quad t_0(A)=\mathbf{ff}=t_1(A)}{(x,t_0 \vee t_1)\in S}
\end{mathpar}
where $i \in \{0,1\}$ and $\neg i = (i+1 \mod{2})$ flips indexes of the arguments.
\begin{proposition}\label{prop:Boolean:Ext:sound}
    The extension $\mathit{bool}$ is sound.
\end{proposition}

To see that the extension $\mathit{bool}$ formally internalizes the algebraic properties of Boolean expressions, denote by $t \cong_{A,R} t'$ the equality 
of the set of dependencies $\{ x \mid (x,t) \in \mathit{bool}(A,R) \} = \{ x \mid (x,t') \in \mathit{bool}(A,R) \}$.
Then, one obtains 
\begin{mathpar}
t \wedge \mathbf{ff} \cong_{A,R} \mathbf{ff}
\and
t \vee \mathbf{tt} \cong_{A,R} \mathbf{tt}
\and
t \wedge \mathbf{tt} \cong_{A,R} t
\and
t \vee \mathbf{ff} \cong_{A,R} t \,.
\end{mathpar}
In fact, the following more general inferences are valid
\begin{equation*}
\frac{t'(A) = \mathbf{ff}
\quad t' \cong_{A,R} \mathbf{ff}
    }{ t \wedge t' \cong_{A,R} \mathbf{ff} }
\;\;\;\;
\frac{t'(A) = \mathbf{tt}
    }{ t \vee t' \cong_{A,R} \mathbf{tt} }
\;\;\;\;
\frac{t'(A) = \mathbf{tt}
    }{ t \wedge t' \cong_{A,R} t }
\;\;\;\;
\frac{t'(A) = \mathbf{ff}
\quad t' \cong_{A,R} \mathbf{ff}
    }{ t \vee t' \cong_{A,R} t }
\end{equation*}
The asymmetry in the assumptions arises because $t'(A) = \mathbf{tt}$ directly implies $t' \cong_{A,R} \mathbf{tt}$. In contrast, when $t'(A) = \mathbf{ff}$, assuming $t' \cong_{A,R} \mathbf{ff}$ is necessary to ensure that $t'(\mu F) = t'(A) = \mathbf{ff}$ (by \autoref{improved:lemmaGR} and \autoref{prop:Boolean:Ext:sound}).

\subsection{Local Extensions}

Also local oracles can be induced from extensions, but their soundness needs to be adapted. To this end, define the relation
\[
{\hflows_A} = \{(x,t) \mid \exists \pi, \pi' \in \V^* \text{ such that } 
t \big( F_{\pi \pi'} (A) \big) \neq t \big( F_{\pi x \pi'} (A) \big) \}.
\]

\begin{definition}
An extension $\lift$ is \emph{locally sound} if $\flows_A \subseteq R$ implies $\hflows_A \subseteq \lift(A, R)$, for every $A\in\A$ and $R\subseteq \V \times \V$. 
The local oracle induced by $\lift$ is
\begin{align*}
\ora[\lift](V,A,R) = {}
&\big( \V \times V^c \big) \cup 
    \{(x,y) \mid y \in V, \ x \neq y, \ (x,t_y) \in \lift(A,R) \} \cup {} 
\\
&\{ (x,x) \mid x \in V, \big( A(x) \neq t_x(A) \text{ or } \exists z . (z,t_x) \in \lift(A,R) \big) \} \,.
\end{align*}
\end{definition}

\begin{proposition}\label{LOCAL:prop:sound:oracles:from:liftings}
A locally sound extension $\lift$ induces a sound local oracle $\ora[\lift]$.
\end{proposition}

We conclude this section by identifying a natural condition on extensions under which soundness entails local soundness.
\begin{proposition}\label{extension:sound+anti-mon:inAssign:implies:locally:sound}
Let $\lift$ be an extension satisfying $\lift(F_\pi(A),R) \subseteq \lift(A,R)$ for every $A$, $\pi$ and $R$ such that $\flows_A \subseteq R$. 
If $\lift$ is sound, then it is also locally sound.
\end{proposition}

This condition is not only theoretically well-motivated but also satisfied by all the extensions introduced in this work (see~\autoref{list:locally:sound:extensions}).

\section{Experimental Results}
\label{Experimental Results}
We implemented a Java prototype%
\footnote{A virtual machine to reproduce the experiments is available at \url{https://homes.cs.aau.dk/~giovbacci/tools/LocalK/VM.zip}.}
of both the global and local algorithms. The implementation provides several built-in oracles for both local and global solvers, which can be combined as shown in Section \ref{sec:oracles} to obtain finer oracles.

\paragraph{Comparison with state-of-the-art.}
We consider a number of problems, including %
bisimulation checking and model-checking.
We show experimental evidence that our approach is competitive with both domain-specific implementations presented in~\cite{%
AndersenAEHLOSW15,JensenLSO16} and the general framework of ADGs~\cite{EnevoldsenLS19}.

For each benchmark, we report the average execution time over ten runs. 
We enforced a timeout policy, interrupting the execution after 25 minutes (indicated as TO in the tables). All experiments were conducted on a Xeon W-1250P 4.10 GHz processor with 16GB RAM and a 512GB SSD, running Linux Ubuntu.

\setlength{\tabcolsep}{3pt}
\begin{table}[t]
\centering
\scriptsize

\begin{tabular}{|c|c|c|c|}
    \hline
    \multirow{2}{*}{B} & \multicolumn{3}{c|}{Time (s)} \\ \cline{2-4}
    & \texttt{ADG} & \texttt{LocalK} & \texttt{CAAL}\\ \hline
    \multicolumn{4}{|c|}{Lossy ABP -- nonbisimilar} \\ \hline
    4 & 10.25 & \textbf{0.12} & 50.79 \\ \hline
    5 & 9.49 & \textbf{0.65} & TO\\ \hline
    6 & 10.77 & \textbf{3.26} & TO\\ \hline\hline
    \multicolumn{4}{|c|}{Lossy ABP -- bisimilar} \\ \hline
    2 & 8.58 & 2.54 & \textbf{2.02}\\ \hline
    3 & \textbf{48.23} & 417.03 & 113.00 \\ \hline
    4 & \textbf{1263.2} & TO & TO\\ \hline
\end{tabular}
\begin{tabular}{|c|c|c|c|}
    \hline
    \multirow{2}{*}{N} & \multicolumn{3}{c|}{Time (s)} \\ \cline{2-4}
    & \texttt{ADG} & \texttt{LocalK} & \texttt{CAAL}\\ \hline
    \multicolumn{4}{|c|}{Leader -- nonbisimilar} \\ \hline
    8 & 8.54 & \textbf{0.68} & 3.62 \\ \hline
    9 & 14.83 & \textbf{5.14} & 50.23 \\ \hline
    10 & 68.82 & \textbf{42.81} & 1317.05\\ \hline \hline
    \multicolumn{4}{|c|}{Leader -- bisimilar} \\ \hline
    8 & 17.06 & \textbf{2.44} & 7.67\\ \hline
    9 & 100.25 & \textbf{29.24} & 116.39 \\ \hline
    10 & 954.87 & \textbf{332.80} &  TO \\ \hline
\end{tabular}
\begin{tabular}{|c|c|c|c|c|c|}
    \hline
    \multicolumn{3}{|c|}{\multirow{2}{*}{Param}} & \multicolumn{2}{c|}{Time (s)} \\ \cline{4-5}
    \multicolumn{3}{|c|}{} & \texttt{~~WKTool~~} & \texttt{LocalK} \\ \hline
    B & M & X & \multicolumn{2}{c|}{ABP -- $\text{EF}[{\leq}X] (\mathit{del} {=} M)$} \\ \hline
    5 & 7 & 35 & 3.492 & \textbf{0.240} \\ \hline
    5 & 8 & 40 & 2.064 & \textbf{0.734} \\ \hline
    6 & 5 & 30 & 3.831 & \textbf{3.806} \\ \hline
    \hline
    \multicolumn{2}{|c|}{N} & X & 
    \multicolumn{2}{c|}{Leader -- $\text{EF}[{\leq}X]\mathit{leader}$} \\ \hline
    \multicolumn{2}{|c|}{10} & 10 & 2.202 & \textbf{1.451} \\ \hline
    \multicolumn{2}{|c|}{11} & 11 & 10.029 & \textbf{3.377} \\ \hline
    \multicolumn{2}{|c|}{12} & 12 & 42.841 & \textbf{10.778} \\ \hline
\end{tabular}
\vspace{2ex}
\caption{Comparisons with ADG~\cite{EnevoldsenLS19}, CAAL~\cite{AndersenAEHLOSW15}, and \texttt{WKTool}~\cite{JensenLSO16}. The parameters $B$, $N$, $M$, and $X$ respectively indicate: the buffer size, the number of processes, the number of messages, and the weight bound in the query $\text{EF}[{\leq} X]\varphi$.   
}
\label{tab:comparison1}
\end{table}

\begin{description}[fullwidth, itemsep=1ex, topsep=1ex]

\item[Bisimulation checking.]
We encoded weak bisimulation checking of CCS processes as a Boolean equation system in our framework and compared the performance of our local algorithm with the ADG implementation from~\cite{ELS22} and CAAL~\cite{AndersenAEHLOSW15} on benchmarks from~\cite{ELS22,dalsgaard2018distributed}. 
Like~\cite{dalsgaard2018distributed}, our encoding follows~\cite{10.1007/978-3-319-47677-3_13}.
The experiments used both the $\mathit{bool}$ extension-induced oracle and $\mathit{smax}$, though only results for $\mathit{smax}$ are reported due to negligible differences. 
\autoref{tab:comparison1} shows that our local algorithm often greatly outperformed both ADG and CAAL. %
An exception was the bisimilar-ABP (Alternating Bit Protocol) benchmarks, where our oracles failed to prune the state space effectively, and over $95\%$ of the total runtime was spent generating the equations system.

\item[Model-Checking WCTL.]
We encoded model checking of weighted CTL over weighted Kripke structures following the approach of~\cite{JensenLSO16}, using equation systems over the poset $(\mathbb{N} \cup {\infty}, \geq)$.
We compared the execution times of our local algorithm using the $\mathit{smax}$ oracle with WKTool %
on a subset of benchmarks from~\cite{JensenLSO16}. 
\autoref{tab:comparison1}(right) shows that \texttt{LocalK} outperforms WKTool on all benchmarks. In particular, on the Leader Election benchmark, we achieve speedups up to $300\%$.

\end{description}

\paragraph{A Comparison of Oracles: precision vs overhead.}
We evaluated the performance of the local algorithm under different local oracles on a list of benchmarks consisting of $1000$ randomly generated Boolean equation systems, each containing $3000$ equations. To emulate conditions in which a local exploration takes place, we introduced an artificial delay of 10 milliseconds for the generation of each equation. For each experiment, we measured the execution time and the number of visited equations $|V|$. 

\autoref{fig:cactus-oracles} reports the data for the local oracles $\mathit{locmax}$, $\mathit{smax}$, $\ora^{\mathit{bool}}$, and $\mathit{comp}$ (that we define below).
The plot to the left shows that $\mathit{locmax}$ and $\mathit{smax}$ are indistinguishable, and they solve all benchmarks within $35$ seconds. Nonetheless, $\ora^{\mathit{bool}}$ solves about $60\%$ of the benchmarks significantly faster than them. This behaviour is explained by looking at the plot to the right: while $\mathit{locmax}$ and $\mathit{smax}$ require exploring the entire system to achieve a $50\%$ coverage, $\ora^{\mathit{bool}}$ achieves $60\%$ coverage by visiting fewer than $80$ variables, yielding substantial savings in the generation of the equations. 
Beyond this threshold, the additional precision of $\ora^{\mathit{bool}}$ is outweighed by its computational overhead. Based on this observation, we define the following composed local oracle, combining the strengths of $\ora^{\mathit{bool}}$ and $\mathit{smax}$:
\begin{equation*}
    \mathit{comp}(V,A,R) = \begin{cases}
        (\ora^{\mathit{bool}} \circ \mathit{smax})(V,A,R) & \text{if $|V| \leq 80$}, \\
        \mathit{smax}(V,A,R) & \text{otherwise}.
    \end{cases}
\end{equation*}
This local oracle is sound by \autoref{prop:oracles:properties}. As shown in \autoref{fig:cactus-oracles},  $\mathit{comp}$ achieves the best runtime performance while remaining conservative in its exploration, showcasing the benefits of our compositional framework.

\begin{figure}[t]
    \centering
    \includegraphics[width=0.47\textwidth]{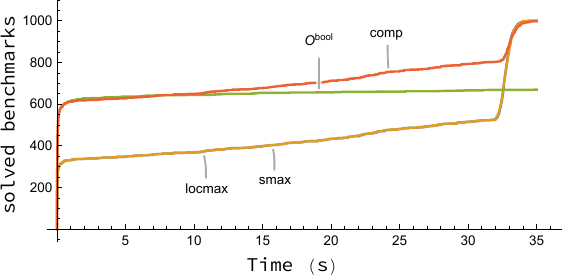}
    \hspace{1ex}
    \includegraphics[width=0.47\textwidth]{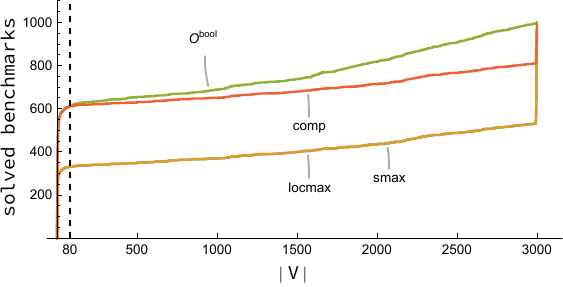}
    \vspace{-10pt}
    \caption{Cactus plots comparing different local oracles w.r.t. the runtime (left) and the number of visited variables $|V|$ (right) during the execution of \proc{LocalK}. } 
    \label{fig:cactus-oracles}
\end{figure}

\paragraph{Global vs Local.}
We conclude with a comparison of \proc{GlobalK} and \proc{LocalK} using different oracles. 
\autoref{tab:localglobal} reports execution time, number of iterations, and number of visited variables on a set of small benchmarks for weak bisimilarity checking. 
The results show that a lazy exploration can drastically reduce the number of variables visited by \proc{LocalK} when a counterexample exists. 
Conversely, \proc{GlobalK} benefits from the full knowledge of the system, often requiring fewer iterations. 
However, as system size increases, the overhead of additional local iterations is typically offset by avoiding the costly construction of the full system of equations.

\begin{table}[t]
    \centering
    \scriptsize
    \begin{tabular}{|c|l|c|c|c|c|c|c|c|c|c|}
    \hline
    &
    \multirow{2}{*}{Model} & \multirow{2}{*}{$|\V|$} & \multirow{2}{*}{$\cong$} & \multicolumn{3}{c|}{\proc{LocalK}} & \multicolumn{4}{c|}{\proc{GlobalK}} \\ \cline{5-11}
   & & & & \text{Time (ms)} & \text{\# Iter} & \text{|V|} & \text{Loop (ms)} & \text{Gen (ms)}& \text{Time (ms)} & \text{\# Iter} \\\hline
 \multirow{5}{*}{\rotatebox[origin=c]{90}{$\mathit{smax}$}} 
 & \text{BasicBuffer} & 24 & yes & 4 & 53 & 24 & 1 & 3 & 4 & 24 \\ \cline{2-11}
 & \text{Peterson} & 148 & no & 4 & 4 & 2 & 1 & 14 & 15 & 3  \\\cline{2-11}
 & \text{Dekker} & 290 & yes & 48 & 2316 & 290 & 1 & 23 & 24 & 290  \\\cline{2-11}
 & \text{ABP(B4)} & 1356 & no & 151 & 4 & 2 & 1 & 6232 & 6233 & 3  \\\cline{2-11}
 & \text{ABP(B3)} & 3568 & yes & 229527 & 8165 & 3568 & 1 & 229003 & 229004 & 3568  \\\hline
 \multirow{5}{*}{\rotatebox[origin=c]{90}{$\ora^{\mathit{bool}}$}}
 & \text{BasicBuffer} & 24 & yes & 4 & 38 & 24 & 1 & 4 & 3 & 24  \\ \cline{2-11}
 & \text{Peterson} & 148 & no & 4 & 4 & 2 & 26 & 12 & 38 & 2  \\ \cline{2-11}
 & \text{Dekker} & 290 & yes & 90 & 532 & 290 & 29 & 26 & 55 & 85  \\ \cline{2-11}
 & \text{ABP(B4)} & 1356 & no & 168 & 4 & 2 & 423 & 6302 & 6725 & 2  \\ \cline{2-11}
 & \text{ABP(B3)} & 3568 & yes & 243058 & 4048 & 3568 & 10406 & 239686 & 294315 & 399  \\ \hline
\end{tabular}
\vspace{1ex}
    \caption{Comparison between local and global algorithms. For \proc{GlobalK}, we also report both the time spent in the main loop and that spent to generate the  system.}
    \label{tab:localglobal}
\end{table}

\section{Conclusions and Future Work}
\label{sec:conclusions}
We introduced a general framework for the analysis of variable dependencies in fixed-point computations, supported by new global and local algorithms and the notion of oracles as a unifying and compositional mechanism for refining dependency information. 
We offer a theoretical framework where soundness emerges through the composition of simpler oracles, and proofs become easier.
Our results show that our theoretical framework retains practical efficiency also when compared with domain-specific tools. Nevertheless, one needs to find a good balance between precision and the computational overhead of the chosen oracle.

Future work will focus on optimizing the implementation to eliminate avoidable overheads, refining strategies for variable extraction from the $\id{todo}$ set, and developing additional classes of domain-specific oracles.
Another research direction is to generalize the framework beyond Noetherian posets, enabling fixed-point computation over continuous domains (e.g., the unit interval $[0,1]$), which are relevant in the verification of quantitative and probabilistic systems.

\begin{credits}
\subsubsection{\ackname} This work was partially supported by the S4OS Villum Investigator Grant nr. 37819 from Villum Fonden.

\subsubsection{\discintname}
The authors have no competing interests to declare that are relevant to the content of this article.
\end{credits}

\bibliography{biblio}
\bibliographystyle{splncs04} %

\appendix

\section{A Wider Selection of Oracles and Local Oracles}
\label{sec:appendix}

We present additional oracles that, due to space limitations, were not included in the main text. These oracles and their local variants further illustrate the flexibility of our framework.

\subsection{Additional Oracles} \label{sec:moreOracles}

Recall that $\args(f)$ denotes the set of arguments of a function $f \colon \D^\V \to \D$, \ie, the smallest set $S \subseteq \V$ such that, there exists $g \colon \D^S \to \D$ with $f(A) = g(A|_S)$ for all $A \in \A$.

\begin{example} \label{ANTICIPATED:ex:argOracle}
    Consider the system of equations from \autoref{ex:locality}. %
    It is easy to see that the arguments of the functions associated to $x,y,z$ are $\args(f_x) = \args(f_y) = \{y,z\}$ and $\args(f_z) = \emptyset$. 
    The only non-trivial fact we focus on here is that $\args(f_w)= \emptyset$, as $f_w$ is the constant function $\mathbf{tt}$. %
\end{example}

\subsubsection{Reachable Arguments Oracle.}
Clearly, the evaluation of $f_x$ can be influenced only by updating the variables among its arguments.
This observation induces an obvious oracle, called \emph{reachable arguments oracle}, defined as follows, for $A \in \A$ and $R \subseteq \V \times \V$
\begin{align*}
    \mathit{args}(A, R) &= \big( \{(x,y) \mid x \in \args(f_y)\} \big)^* \,,
\end{align*}
where $(-)^*$ denotes the reflexive and transitive closure of a relation. 

Intuitively, $(x,y) \in \mathit{args}(A, R)$ if and only if either $x=y$ or there is a sequence of equations $(z_0 = f_{z_0}), \dots, (z_k = f_{z_k})$ in $\E$, such that $x = z_0$, $y = z_k$, and $z_{i} \in \args(f_{ z_{i+1} })$ for all $0 \leq i < k$. Note that the oracle $\mathit{args}$ is a constant map, depending neither on the assignment nor on the given relation.

\begin{example} \label{ex:argOracle}
    Consider the system of equations from \autoref{ex:locality}. 
    Referring to \autoref{ANTICIPATED:ex:argOracle}, 
    for arbitrary $A$ and $R$,
    \begin{align*}
        \mathit{args}(A, R) &= \big( \{(y,x), (z,x), (y,y), (z,y) \} \big)^* 
        \\
        &= \{(x,x), (y,x), (z,x), (y,y), (z,y), (z,z), (w,w) \} \,.
    \end{align*}
    Although simple, the use of this oracle in combination with other oracles provides us with a surprisingly good dependency analysis. For example $(\mathit{max}_l \cap \mathit{max}_r \cap \mathit{args})(F_z(\bot), R) = \cc_{F_z(\bot)}$, meaning that this oracle is perfectly accurate on the assignment $F_z(\bot)$.
\end{example}

\begin{proposition}\label{args:oracle:sound}
    The oracle $\mathit{args}$ is sound.
\end{proposition}
\begin{proof}%
Soundness of $\mathit{args}$ follows by showing that, for all 
$A \in \A$ and $R \subseteq \V \times \V$, if $F_\pi(A)(y) \neq F_\pi F_x(A)(y)$ for a
$\pi \in \V^*$ of minimal length, then $(x,y) \in \mathit{args}(A,R)$. 
We prove this implication by induction on $|\pi| \geq 0$.
\begin{trivlist}
    \item \textsc{Base Case.} $\pi$ is empty. Then $A(y) \neq F_x(A)(y)$, which implies $x = y$. By reflexive closure, $(x,y) \in \mathit{args}(A,R)$.
    \item \textsc{Inductive Step.} $\pi$ is non-empty. By \autoref{remark:finally:used}, we can assume $\pi = \pi'y$. Then our assumption becomes $f_y(F_{\pi'}(A)) \neq f_y(F_{\pi'} F_x(A))$, which in turn implies $F_{\pi'}(A) \neq F_{\pi'} F_x(A)$.
    
    Next we prove that there is $z \in \args(f_y)$ such that
    $F_{\pi'}(A)(z) \neq F_{\pi'} F_x(A)(z)$. By contradiction,
    assume that $F_{\pi'}(A)(z) = F_{\pi'} F_x(A)(z)$ for all $z \in \args(f_y)$, 
    that is, $F_{\pi'}(A)|_{\args(f_y)} = F_{\pi'} F_x(A)|_{\args(f_y)}$. 
    By definition of $\args(f_y)$, there exists $g \colon \D^{\args(f_y)} \to \D$ such that 
    \begin{align*}
        f_y(F_{\pi'}(A)) = g(F_{\pi'}(A)|_{\args(f_y)})
        &&\text{and}&&
        f_y(F_{\pi'} F_x(A)) = g(F_{\pi'} F_x(A)|_{\args(f_y)}) \,.
    \end{align*}
    As we assumed $F_{\pi'}(A)|_{\args(f_y)} = F_{\pi'} F_x(A)|_{\args(f_y)}$, from the above equalities we obtain that $f_y(F_{\pi'}(A)) = f_y(F_{\pi'} F_x(A))$, a contradiction.

    Now, pick $z \in \args(f_y)$ such that $F_{\pi'}(A)(z) \neq F_{\pi'} F_x(A)(z)$ ---the existence of such a variable has been just proven. 
    There is $\pi''$ of minimal length such that $F_{\pi''}(A)(z) \neq F_{\pi''} F_x(A)(z)$, and so $|\pi''| \leq |\pi'| = |\pi|-1$. %
    By inductive hypothesis on $\pi''$, we have $(x,z) \in \mathit{args}(A,R)$; moreover, $z \in \args(f_y)$ yields $(z,y) \in \mathit{args}(A,R)$, so we conclude that $(x,y) \in \mathit{args}(A,R)$.%
\end{trivlist}
\end{proof}

\subsubsection{Dependency and Triggering Oracles.}
The dependency analyses discussed so far have been informally based on the idea that updates to a variable may be triggered by updates to the variables it depends on. The definitions of the following oracles, referred to respectively as the \emph{dependency} and \emph{trigger} oracles, are based on this basic intuitive concept. 
\begin{align*}
    \mathit{dep}(A, R) &= \{ (x,y) \mid (z,x)\in R, \text{ for some } z\in\V \} \,, \\
    \mathit{trg}(A,R) &= \{ (x,y) \mid (z,x)\in R \text{ and } A(z) \neq f_z(A), \text{ for some } z\in\V \} \,. %
\end{align*}
Essentially, $(x,y) \in \mathit{dep}(A,R)$ captures the idea that $x$ may trigger the update of $y$ if the given dependency analysis $R$ indicates that $x$ may change its value due to the update of a third variable $z$. Similarly for $(x,y)\in \mathit{trg}(A,R)$, with the exception that the update on $z$ is tested against the assignment $A$.

\begin{proposition} \label{prop:trigdevsound}
    The oracles $\mathit{dep}$ and $\mathit{trg}$ are sound.
\end{proposition}
\begin{proof}%
We start by proving soundness of $\mathit{trg}$. Let $A \in \A$
and $R \subseteq \V \times \V$ such that $\cc_A \subseteq R$.
We want to show that $\cc_A \subseteq \mathit{trg}(A,R)$. We will do that by proving the contra-variant membership implication. Let $(x,y) \notin \mathit{trg}(A,R)$. 
Note that 
\begin{equation}
    \mathit{trg}(A,R) = 
    \{ (x,y) \mid (z,x) \in (\mathit{ns} \cap \mathit{i})(A,R)
    , \text{ for some } z\in\V \} \,.
    \label{eq:trigalt}
\end{equation}
Therefore, for all $z \in \V$, $(z,x) \notin (\mathit{ns} \cap \mathit{i})(A,R)$. By soundness of the oracle $\mathit{ns} \cap \mathit{i}$, we further have that $(z,x) \notin \cc_A$. By \autoref{th:nowLocalFix}, we get $A(x) = \mu F(x)$, thus $A = F_x(A)$. By definition of $\cc_A$, we conclude that $(x,y) \notin \cc_A$.

The soundness of $\mathit{dep}$ follows from the above and \autoref{prop:oracles:properties}\eqref{upclosure}, after noticing that $\mathit{trg} \preceq \mathit{dep}$.
\end{proof}

We observe that these oracles are better used in combination with other oracles as illustrated in~\autoref{ex:dep&trgOracles}.

\begin{example} \label{ex:dep&trgOracles}
    Consider the system of equations from \autoref{ex:locality} and assume to use as $R = \mathit{args}(\bot, \V\times\V)$ as starting dependency analysis (\cf~\autoref{ex:argOracle}). We examine two oracles: $O_1 = \mathit{trg} \circ \mathit{args}$ and $O_2 = O_1 \cap \mathit{args}$.
    Clearly, the second oracle is is more accurate than the first one.
    \begin{align*}
        \begin{aligned}
            O_1(\bot,R) &= \V \times \V \\
            O_1(F_z(\bot),R) &= \{ x,w \} \times \V \\
            O_1(F_x F_z(\bot),R) &= \{ (w,w) \}
        \end{aligned}
        &&
        \begin{aligned}
            O_2(\bot,R) &= \mathit{args}(\bot, \V\times\V) \\
            O_2(F_z(\bot),R) &= \{ (x,x), (w,w) \} \\
            O_2(F_x F_z(\bot),R) &= \{ (w,w) \} \,.
        \end{aligned}
    \end{align*}
    Observe that these oracles perform very well on this particular system of equations. Indeed, from the second assignment, $O_2$ is as precise as the now relation:
    \begin{align*}
        \cc_{ F_z(\bot) } &= O_2(F_z(\bot),R) \subsetneq O_1(F_z(\bot),R)\\
        \cc_{ F_x F_z(\bot) } &= O_2( F_x F_z(\bot),R) = O_1( F_x F_z(\bot),R) \,.
    \end{align*}
    This demonstrates that by leveraging the specific characteristics of the equation system, one can achieve effective dependency analysis even with computationally inexpensive oracles tailored to the system.
\end{example}

\subsubsection{Deflationary Oracles.}
The next oracles exploit a basic order-theoretic property of the equation system: if $f_x$ is deflationary (\ie, $f_x(A) \sqsubseteq A(x)$ for all $A \in \D^\V$), then the least solution for $x$ must be the bottom element in $\D$.

Building on this observation, we define the \emph{deflationary oracles} with left and
right variant, as follows:
\begin{align*}
    \mathit{dfl}_l(A,R) = \{ (x,y) \mid \text{$f_x$ not deflationary} \} \,, \\
    \mathit{dfl}_r(A,R) = \{ (x,y) \mid \text{$f_y$ not deflationary} \}.
\end{align*}

\begin{proposition}\label{deflationary:oracles:soundness}
    The oracles $\mathit{dfl}_l$ and $\mathit{dfl}_r$ are sound.
\end{proposition}
\begin{proof}%
We start by proving that the relation $\mathit{dfl}_l(A,R)$ satisfies the following:
\[
\text{if $(x,y) \notin \mathit{dfl}_l(A,R)$, then $\mu F(x) = \bot_\D$.}
\]
Let $(x,y) \notin \mathit{dfl}_l(A,R)$, and consider an arbitrary $B \in \A$. 
By~\autoref{th:partialKleene}\eqref{Kleene:intro:item1}, we have that $B \sqsubseteq F_x(B)$, so $(x,y) \notin \mathit{dfl}_l(A,R)$ implies that $B(x) = f_x(B)$, \ie, $B = F_x(B)$. In other words, every $B \in \A$ is a fixed point of $F_x$, and then one can easily verify that $\mu F(x) = \bot_\D$. 

To show that $\cc_A \subseteq \mathit{dfl}_l(A,R)$, let $(x,y) \in \cc_A$. Then $A \neq F_x(A)$ by the definition of $\cc_A$, and so $A \sqsubset F_x(A) \sqsubseteq \mu F$ by~\autoref{th:partialKleene}\eqref{Kleene:intro:item1}. In particular, $\bot \sqsubset \mu F(x)$, and then $(x,y) \in \mathit{dfl}_l(A,R)$ by the first part of the proof. 

The soundness of $\mathit{dfl}_r$ follows similarly, using the fact that $(x,y) \notin \mathit{dfl}_r(A,R)$ implies $\mu F(y) = \bot_\D$.
\end{proof}

\begin{example}
Consider the system of equations in \autoref{ex:locality}. The function $f_y$ is clearly deflationary, as $f_y(A) = A(y) \wedge A(z) \sqsubseteq A(y)$. Consequently, the oracle $\mathit{defl}_l \cap \mathit{defl}_r$ eliminates all occurrences of the variable $y$ from every dependency analysis, including the trivial one $\V \times \V$.
As a result, running $\proc{GlobalK}$ with $y$ as the target variable would terminate immediately, even before the iteration begins (\cf~concluding remarks of \autoref{ex:locality}).
\end{example}

\begin{remark}
We note that the definition of the deflationary oracles assumes the user is aware of which functions are deflationary in advance, as verifying this property directly would be too costly.

Furthermore, while some functions may not be inherently deflationary, they can exhibit deflationary behavior when their definitions are unfolded. For instance,
in the system of equations below
\begin{align*}
    x = y \wedge z \,, &&
    y = x \wedge w \,, &&
    w = \dots
\end{align*}
$f_x$ is not deflationary, but $f_x \circ F_y$ is. Our deflationary oracles do not cover these cases.
\end{remark}

\subsection{Additional Extensions}
\label{sec:additional:extensions}

For an extension $\lift \colon \A \times \wp(\V \times \V) \to \wp(\V\times \T)$, and $(A,R) \in \A \times \wp(\V \times \V)$, it is convenient to describe the relation $\lift(A,R)$ through its preimages, \ie, for every term $t\in\T$ we define the set $\lift(A,R)(t) = \{ x \in \V \mid (x,t) \in \lift(A,R) \}$ of the dependencies of $t$. This will be formally carried out by (equivalently) considering the extensions as functions $\lift \colon \A \times \wp(\V \times \V) \to (\T \to \wp(\V))$.

\subsubsection{Reachable Variables Extension.}
Next we revisit the reachable arguments oracle $\mathit{args}$ from the point of view of extensions. For a term $t$, the main difference will be to use variables $\var(t)$, a syntactic concept, instead of arguments $\args(t)$, a semantic concept%
\footnote{For example, for a variable $x$, the Boolean term $t = x \vee \mathbf{tt}$ has $\args(t) = \emptyset$ and $\var(t) = \{x\}$ (see also \autoref{ANTICIPATED:ex:argOracle}). Moreover, $\args(t) \subseteq \var(t)$ holds for every term~$t$.}.

The proposed extension, called \emph{variable extension}, is defined as
\begin{equation*}
    \mathit{vars}(A,R)(t) =
    \{ x \mid (x,z) \in R \text{ for some } z\in\var(t) \}\,.
\end{equation*}

\begin{proposition} \label{prop:vars:ext:sound}
The extension $\mathit{vars}$ is sound.
\end{proposition}
\begin{proof}%
Let $A \in\A$ and $R \subseteq \V\times\V$ such that 
$\cc_A \subseteq R$.
We prove by induction on the structure of $t\in\T$ that $(x,t) \in \hc_A$ implies $x \in \mathit{vars}(A,R)(t)$.
\begin{trivlist}
\item \textsc{Base Case} ($t = y \in \V$).
Assume $(x,y) \in \hc_A$. Then we have
\begin{align*}
(x,y) \in \hc_A
    &\iff (x,y) \in \cc_A \tag{def. $\hc_A$} \\
    &\implies (x,y) \in R \tag{$\cc_A \subseteq R$} \\
    &\iff x \in \mathit{vars}(A,R)(y) \tag{def. $\mathit{vars}$}
\end{align*}
\item \textsc{Inductive Step} ($t = \sigma(t_1,\dots,t_n)$). 
Assume $(x, t) \in \hc_A$. Then we have
\begin{align*}
    (x, t) \in \hc_A &\iff \exists \pi'.\, t(F_{\pi'}(A)) \neq t(F_{x\pi'}(A)) \tag{def. $\hc_A$}\\
    &\implies \exists i \in \{1\dots n \} .\, \exists \pi'.\, t_i(F_{\pi'}(A)) \neq t_i(F_{x\pi'}(A)) \tag{$\sigma$ function} \\
    &\iff \exists i \in \{1\dots n \} .\, (x, t_i) \in \hc_A \tag{def. $\hc_A$} \\
    &\implies \exists i \in \{1\dots n \} .\, x\in \mathit{vars}(A,R)(t_i) 
    \tag{ind. hp} \\
     &\iff \exists i \in \{1\dots n \} .\, \exists y \in \var(t_i) .\, (x, y) \in R 
    \tag{def. $\mathit{vars}$} \\
    &\iff \exists y \in \var(t) .\, (x, y) \in R 
    \tag{$\var(t) = \bigcup_i \var(t_i)$} \\
    &\iff x\in \mathit{vars}(A,R)(t) 
    \tag{def. $\mathit{vars}$}
\end{align*}
\end{trivlist}  
\end{proof}

\subsubsection{Maximality Extension.}
We revisit maximality oracles through the lens of extensions, now allowing the evaluation of sub-terms in right-hand sides of the equations. This yields a more precise dependency analysis. This is achieved naturally through the following extension, which we call \emph{maximality extension}, defined inductively on terms as follows
\begin{align*}
\mathit{maxExt}(A,R)(x) 
    &= \{ y \mid \text{$A(x)$ and $A(y)$ not maximal}\} \,, \\
\mathit{maxExt}(A,R)(\sigma(t_1,\twodots,t_n))
    &= \left\{ y \;\left|\;
    \begin{array}{c}
         \text{$\sigma(t_1,\twodots,t_n)(A)$ not maximal}\\
         \text{and }\exists i.~
         y \in \mathit{maxExt}(A,R)(t_i)
    \end{array}
    \right.\right\} \,.
\end{align*}

Next we formally establish what was announced: 
the oracle induced by $\mathit{maxExt}$ is more accurate than each of the maximality oracles, and in fact strictly more accurate (see \autoref{ex:oracle:maxExt:better:than:maxl:and:R}).

\begin{proposition}\label{maximality:extension:sound}
The extension $\mathit{maxExt}$ is sound, and $\ora[\mathit{maxExt}] \preceq \mathit{max}_l \cap \mathit{max}_r$.
\end{proposition}
\begin{proof}%
We first prove that the extension $\mathit{maxExt}$ is sound.

Let $A\in\A$ and $R \subseteq \V\times\V$. It suffices to verify that $y\notin\mathit{maxExt}(A,R)(t)$ implies $(y,t)\notin\hc_A$, for all $y\in\V$ and $t\in\T$.
We proceed by induction on $t$.
\begin{trivlist}
\item \textsc{Base Case} ($t = x\in\V$).
Assume $y\notin\mathit{maxExt}(A,R)(x)$. This can happen in the following two cases:
\begin{itemize}
\item $A(x)$ is maximal. Then by \autoref{th:partialKleene} $A(x) = \mu F(x)$, which in turn gives $(y,x)\notin\cc_A$ by \autoref{th:nowLocalFix}.

\item $A(y)$ maximal. Then by \autoref{th:partialKleene} 
$A(y) = \mu F (y)$ which implies $(y,x)\notin\cc_A$ by definition of $\cc_A$.
\end{itemize}
In both cases we obtained that $(y,x)\notin\cc_A$. As $\hc_A$ coincides with $\cc_A$ on variables, we have $(y,x)\notin\hc_A$.

\item \textsc{Inductive Step} ($t = \sigma(t_1,\twodots,t_n)$). Let $y\notin\mathit{maxExt}(A,R)(t)$. %
By definition of $\mathit{maxExt}$ extension, only two cases are possible:
\begin{enumerate}
\item\label{item1:sound:max:ext} $t(A)$ is maximal. Then, as $A\sqsubseteq\mu F$ (\autoref{th:partialKleene}) and since $t$ has monotonic interpretation, by maximality we get $t(A) = t(\mu F)$. By \autoref{improved:lemmaGR}, we deduce that $(y,t) \notin \hc_A$.

\item %

For all $i \in \{1,\dots,n\}$ we have $y\notin\mathit{maxExt}(A,R)(t_i)$. Then for all $i \in \{1,\dots,n\}$ we deduce $(y,t_i) \notin\hc_A$ by inductive hypothesis. As
\begin{align*}
    (y,t)\in\hc_A 
    &\iff \exists \pi. t(F_\pi(A)) \neq t(F_{y\pi}(A)) \tag{def.~$\hc_A$}\\
    &\implies \exists i \in \{1,\dots,n\}. 
        \exists \pi. t_i(F_\pi(A)) \neq t_i(F_{y\pi}(A)) 
        \tag{$\sigma$ function}\\
    &\iff \exists i\in \{1,\dots,n\}. (y,t_i)\in\hc_A
     \tag{def.~$\hc_A$}
\end{align*}
it must be that $(y,t)\notin\hc_A$.%
\end{enumerate}
\end{trivlist}

To conclude, we show that $\ora[\mathit{maxExt}] \preceq \mathit{max}_l \cap \mathit{max}_r$.

Fix $A \in \A$ and $R \subseteq \V \times \V$, and let $(x,y) \in \ora[\mathit{maxExt}](A,R)$.
If $x=y$, then $A(x) \neq t_x(A)$, and so $(x,x) \in \mathit{max}_l \cap \mathit{max}_r$. Now assume $x\neq y$. As $x \in \mathit{maxExt}(A,R)(t_y)$, the definition of $\mathit{maxExt}$ implies that:
\begin{itemize}
\item $t_y(A)$ is not maximal, and so $(x,y) \in \mathit{max}_r(A,R)$;

\item $A(x)$ is not maximal (this is easily verified by induction on $t_y$), and so $(x,y) \in \mathit{max}_l(A,R)$.%
\end{itemize}
\end{proof}

\begin{example}\label{ex:oracle:maxExt:better:than:maxl:and:R}
We compare the maximality oracles %
with the oracle induced by the $\mathit{maxExt}$ extension. Consider the following system of equations on the Boolean lattice $\{\mathbf{ff},\mathbf{tt}\}$, with $\wedge$ and $\vee$
interpreted as expected:
\begin{align*}
    x = y \wedge (z \vee w) 
    &&
    y = \mathbf{ff}
    &&
    z = \mathbf{tt}
    &&
    w = \mathbf{tt}
\end{align*}
For the assignment $A = (x\colon \mathbf{ff}; y\colon \mathbf{ff}; z\colon \mathbf{tt}; w\colon \mathbf{ff})$
and a preliminary dependency analysis $R = \mathit{args}(A,\V\times\V) = \{ (x,x),(y,x),(z,x),(w,x),(y,y),(w,w)\}$ we obtain the following refined analysis
\begin{align*}
    (\mathit{max}\cap\mathit{i})(A,R) 
    &= \{(x,x),(y,x),(w,x),(y,y),(w,w)\} \\
    (\mathit{max}_l \cap \mathit{max}_r \cap\mathit{i})(A,R) 
    &= \{(y,x),(w,x),(w,w)\} \\
    \ora[\mathit{maxExt}\cap\mathit{i}](A,R) 
    &= \{(y,x),(w,w)\}.
\end{align*}
As the sub-term $z\vee w$ of $t_x$ evaluates to $(z\vee w)(A) = \mathbf{tt}$, the maximal element, both $z$ and $w$ have no influence on $x$ after reaching the assignment $A$.
However, neither $\mathit{max}$ nor $\mathit{max}_l \cap \mathit{max}_r$ could fully spot this independence.
\end{example}

We conclude by observing that, since extensions are only applicable to algebraic systems of equations, the maximality oracles introduced before %
remain relevant.

\subsection{Additional Local Oracles} \label{sec:more:Local:Oracles}

The local oracles presented here are obtained as the local variants of the oracles from Section~\ref{sec:moreOracles}.

\subsubsection{Local Maximality Oracles.}
We present the local variants of the maximality oracles discussed in the main text. 

We start by recalling definition of the local variant of $\mathit{max}_r$:
\begin{align*}
\mathit{locmax}_r(V, A, R) = {}
&\big( \V \times V^c \big) \cup 
\{ (x,y) \mid y \in V, \ x \neq y, \text{ $f_y(A)$ not maximal} \}\cup {} 
\\
&\{ (x,x) \mid x \in V, \big( A(x) \neq f_x(A) \text{ or } f_x(A) \text{ not maximal} \big) \}  \,.
\end{align*}

\begin{proof}[of \autoref{locmaxr:sound}]
Fix $A \in \A$ and $R \subseteq \V \times \V$ such that $\flows_A \subseteq R$. Let $(x,y) \in \flows_A$, for which we verify that $(x,y) \in \mathit{locmax}_r(V,A, R)$.

If $y \notin V$, then $(x,y) \in \mathit{locmax}_r(V,A, R)$, so assume $y \in V$. 
From $(x,y) \in \flows_A$ and Proposition~\ref{flows=unionOFfutureNows}, there exists $\pi\in \V^*$ %
such that $(x,y) \in \cc_{F_\pi(A)}$. 

If $x \neq y$, then Proposition~\ref{prop:pairsINflows}\eqref{prop:pairsINflows:neq} yields $f_y ( F_{\pi'}( F_\pi(A) ) ) \neq f_y ( F_{\pi'} F_x( F_\pi(A) ) )$ for some $\pi'\in \V^*$. Then, \autoref{th:partialKleene}\eqref{Kleene:intro:item1} gives
\[
f_y ( A ) \sqsubseteq f_y ( F_{\pi'}( F_\pi(A) ) ) \sqsubset f_y ( F_{\pi'} F_x( F_\pi(A) ) ),
\]
so $f_y ( A )$ is not maximal.

Now consider the case when $x = y \in V$. If $A(x) \neq f_x(A)$, then $(x,x) \in \mathit{locmax}_r(V,A, R)$. Finally, assume that $A(x) = f_x(A)$.
As $(x,x) \in \cc_{F_\pi(A)}$, Proposition~\ref{prop:pairsINflows}\eqref{prop:pairsINflows:=} implies that $F_\pi(A)(x) \neq f_x ( F_\pi(A) )$. Then
\[
f_x ( A ) = A(x) \sqsubseteq F_\pi(A)(x) \sqsubset f_x ( F_\pi(A) ),
\]
so $f_x ( A )$ is not maximal, and then $(x,x) \in \mathit{locmax}_r(V,A, R)$.
\end{proof}

A similar argument gives the following sound local variant of $\mathit{max}_l$:
\begin{align*}
\mathit{locmax}_l(V, A, R) = {}
&\big( \V \times V^c \big) \cup \{ (x,x) \mid x \in V \} \cup {} \\
&\{ (x,y) \mid x \neq y \text{ and $A(x)$ not maximal} \} \,.
\end{align*}

\subsubsection{Local Deflationary Oracles.}
The deflationary oracles $\mathit{dfl}_l$, $\mathit{dfl}_r$ and their intersection  $\mathit{dfl} = \mathit{dfl}_l \cap \mathit{dfl}_r$ require to test whether the right-hand sides of the equations are deflationary functions. As this requires knowledge of the equations, it is relevant to offer their local variants as follows:

\begin{align*}
\mathit{locdfl}_l(V,A,R) &= \big( V^c \times \V \big) \cup \{ (x,y) \mid x \in V \text{ and $f_x$ not deflationary} \} ,
\\
\mathit{locdfl}_r(V,A,R) &= \big( \V \times V^c \big) \cup \{ (x,y) \mid y \in V \text{ and $f_y$ not deflationary} \} ,
\\
\mathit{locdfl}(V,A,R) &= \big( V^c \cup \{ x \in V \mid \text{$f_x$ not deflationary} \} \big)^2 .
\end{align*} 

\begin{proposition}%
The local oracles $\mathit{locdfl}_l$, $\mathit{locdfl}_r$, and $\mathit{locdfl}$ are sound.
\end{proposition}
\begin{proof}%
The local oracles $\mathit{locdfl}_l$ and $\mathit{locdfl}_r$ are sound by a straightforward modification of the proof of Proposition~\ref{deflationary:oracles:soundness}. 
As $\mathit{locdfl} = \mathit{locdfl}_l \cap \mathit{locdfl}_r$, also $\mathit{locdfl}$ is sound by \autoref{prop:oracles:properties}(\ref{operation-closure}). 
\end{proof}

\subsubsection{Local Reachable Arguments Oracle.}
Although we could define a local variant for
the oracle $\mathit{args}$, such an oracle would not be meaningful in practice. This is not a limitation of our framework but rather a consequence of the fact that, by its nature, the local algorithm already governs and restricts the exploration of the system. Hence, a further restriction to the reachable arguments would be redundant.

\subsubsection{Local Triggering Oracle.}
The following is the proposed local variant of $\mathit{trg}$
\begin{align*}
\mathit{loctrg}(V,A,R)&= \{ (x,y) \mid (z,x)\in R \text{ and } A(z) \neq f_z(A), \text{ for some } z\in V \} \ \cup \\
    &\phantom{{}={}} \{ (x,y) \mid (z,x)\in R \text{ for some } z\in \V \setminus V \} .
\end{align*}

\begin{proposition}\label{loctrig:sound}
The local oracle $\mathit{loctrg}$ is sound.
\end{proposition}
\begin{proof}%
Fix $A \in \A$ and $R \subseteq \V \times \V$ such that $\flows_A \subseteq R$. Consider $(x,y) \in \flows_A$, for which we check that $(x,y) \in \mathit{loctrg}(V,A,R)$.

In view of Theorem \ref{th:flowsLocalFix}, the assumption $(x,y) \in \flows_A$ yields that $A(x) \neq \mu F(x)$, and then \autoref{th:nowLocalFix} implies that there is $z \in \V$ such that $(z,x) \in \cc_A$. As $\cc_A \subseteq \flows_A \subseteq R$ we get that $(z,x) \in R$.

If $z\in \V \setminus V$, then $z$ witnesses that $(x,y) \in \mathit{loctrg}(V,A,R)$. If $z \in V$, the previously established $(z,x) \in \cc_A$ also implies that $A(z) \neq f_z(A)$ by definition of $\cc_A$, and so also in this case $z$ witnesses that $(x,y) \in \mathit{loctrg}(V,A,R)$.
\end{proof}

Similarly to its non-local counterpart, also $\mathit{loctrg}$ is better used in combinations with other oracles.

\subsubsection{Local Dependency and Simplified Maximality Oracles.}
\label{Local Dependency and Simplified Maximality Oracles.}
Here we present the local variant of the oracles $\mathit{dep}$ and $\mathit{smax}$.

\begin{definition}\label{def:weakly:anti-monotone}
For a poset $(X, \leq)$, call a function $\phi \colon \A \to X$ weakly anti-monotone if $\phi( F_\pi(A) ) \leq \phi(A)$ for every $A \in \A$ and $\pi \in \V^*$.
\end{definition} 
The terminology is explained by the fact that anti-monotone functions on $(\A , \sqsubseteq)$ are weakly anti-monotone, as $A  \sqsubseteq F_\pi(A)$ for every $A \in \A$ and $\pi \in \V^*$ by~\autoref{th:partialKleene}\eqref{Kleene:intro:item1}.
\begin{lemma}\label{easy:condition:sound:local}
Let $\ora$ be an oracle such that $\ora(\cdot, R)$ is weakly anti-monotone for every $R$. 
Define the local oracle $\ora'$ as 
\[\ora'(V,A,R) = \ora(A,R)\,. \] 
Then, if $\ora$ is sound, so is $\ora'$.
\end{lemma}
\begin{proof}%
Fix $V \subseteq \V$, $A \in \A$, and $R \subseteq \V \times \V$, and assume that $\flows_A\subseteq R$. We have to verify that $\flows_A  \subseteq \ora'(V,A, R)$, \ie, that $\flows_A \subseteq \ora(A, R)$. 
By Proposition \ref{flows=unionOFfutureNows}, we have that $\flows_A = \bigcup_{\pi \in \V^*} \cc_{F_\pi(A)}$, so fix $\pi \in \V^*$, for which we have to check that $\cc_{F_\pi(A)} \subseteq \ora(A, R)$. 

As $\ora$ is sound we have $\cc_{F_\pi(A)} \subseteq \ora(F_\pi(A), R)$. Moreover, $A \sqsubseteq F_\pi(A)$ by~\autoref{th:partialKleene}\eqref{Kleene:intro:item1}, so as $\ora(\cdot,R)$ is weakly anti-monotone, $\ora(F_\pi(A), R) \subseteq \ora(A, R)$.
\end{proof}

We apply the previous result to obtain sound local variants of the oracles $\mathit{dep}$ and $\mathit{smax}$, as follows:
\begin{corollary}\label{locdep:locsmax:sound}
The following local oracles are sound:
\begin{align*}
    \mathit{locdep}(V, A, R) &= \mathit{dep}(A, R) = \{ (x,y) \mid (z,x)\in R, \text{ for some } z\in\V \},
    \\
    \mathit{locsmax}(V,A,R) &= \mathit{smax}(A, R) = \{ (x,y) \mid \text{$A(x)$ and $A(y)$ not maximal} \}.
\end{align*}
\end{corollary}
\begin{proof}%
The oracle $\mathit{dep}$ is sound and independent of the assignment, so \autoref{easy:condition:sound:local} applies immediately. It is also straightforward to verify that $\mathit{smax}$ satisfies the assumptions of \autoref{easy:condition:sound:local}. Thus $\mathit{locdep}$ and $\mathit{locsmax}$ are locally sound.
\end{proof}

\subsubsection{An account of locally sound extensions.}
Here we prove that all the extensions presented in this work are locally sound. As already anticipated in Section~\ref{sec:local:oracles}, this is a consequence of the fact that all these extensions satisfy the conditions of \autoref{extension:sound+anti-mon:inAssign:implies:locally:sound}.

\begin{proposition}\label{list:locally:sound:extensions}
The extensions $\mathit{vars}$, $\mathit{maxExt}$, and $\mathit{bool}$ are locally sound.
\end{proposition}
\begin{proof}%

The reachable variables extension $\mathit{vars}$ does not depend on the assignment, it is sound by~\autoref{prop:vars:ext:sound}, so it is locally sound by~\autoref{extension:sound+anti-mon:inAssign:implies:locally:sound}.

Regarding $\mathit{maxExt}$, we claim that $\mathit{maxExt}( \cdot ,R)$ is weakly anti-monotone for every $R \subseteq \V \times \V$. Indeed, let $A \in \A$, $\pi\in \V^*$, $R \subseteq \V \times \V$, and $t \in \T$. To verify that $\mathit{maxExt}(F_\pi(A),R)(t) \subseteq \mathit{maxExt}(A,R)(t)$, use an easy inductive argument based on the fact that $A \sqsubseteq F_\pi(A)$ by~\autoref{th:partialKleene}\eqref{Kleene:intro:item1}, and so $t(A) \sqsubseteq t(F_\pi(A))$ as the interpretation of $t$ is monotone. Then if $t(F_\pi(A))$ is not maximal, neither $t(A)$ is maximal; this proves our claim. As $\mathit{maxExt}$ is sound by~\autoref{maximality:extension:sound}, it is locally sound by~\autoref{extension:sound+anti-mon:inAssign:implies:locally:sound}. 

To prove that $\mathit{bool}$ is locally sound, one could use an argument similar to the one used in the proof of~\autoref{prop:Boolean:Ext:sound}. We offer instead a different argument based on the weak anti-monotonicity property of $\mathit{bool}(\cdot,R)$ stated in~\autoref{bool:anti-monotone:in:assignment}.
Then, as $\mathit{bool}$ is sound by~\autoref{prop:Boolean:Ext:sound}, we deduce that $\mathit{bool}$ is locally sound using~\autoref{extension:sound+anti-mon:inAssign:implies:locally:sound}.
\end{proof}

\begin{lemma}\label{bool:anti-monotone:in:assignment}
Let $A \in \A$, $\pi\in \V^*$, and $t \in \T$.
\begin{itemize}
\item If $R \subseteq \V \times \V$ is such that $\cc_A \subseteq R$, then $\mathit{bool}(F_\pi(A),R)(t) \subseteq \mathit{bool}(A,R)(t)$.

\item 
If $R \subseteq R' \subseteq \V \times \V$, then $\mathit{bool}(A,R)(t) \subseteq \mathit{bool}(A,R')(t)$. %
\end{itemize}
\end{lemma}
\begin{proof}
Fix $R \subseteq \V \times \V$ such that $\cc_A \subseteq R$. We first prove that $\mathit{bool}(F_\pi(A),R)(t) \subseteq \mathit{bool}(A,R)(t)$ by structural induction on $t \in \T$.
\begin{trivlist}
\item \textsc{Base Case} ($t = y \in \V$). 
As $A(y) \sqsubseteq F_\pi (A)(y)$ by \autoref{th:partialKleene}\eqref{Kleene:intro:item1}, if $F_\pi (A)(y) = \mathbf{ff}$, then $A(y) = \mathbf{ff}$.
As a consequence, 
\begin{align*}
\mathit{bool}(F_\pi(A),R)(y) &= \{ x \mid (x,y) \in R \text{ and } F_\pi (A)(y) = \mathbf{ff} \}\\
&\subseteq \{ x \mid (x,y) \in R \text{ and } A(y) = \mathbf{ff} \}=\mathit{bool}(A,R)(y).
\end{align*} 

\item \textsc{Base Case} (either $t = \mathbf{tt}$ or $t = \mathbf{ff}$).
Then the definition of $\mathit{bool}$ gives us $\mathit{bool}(A,R)(t) = \emptyset = \mathit{bool}(F_\pi(A),R)(t)$.

\item \textsc{Inductive Step} ($t = t_0 \vee t_1$). Let $x \in \mathit{bool}(F_\pi(A),R)(t_0 \vee t_1)$. By definition, there is $i \in\{0,1\}$ such that $x \in \mathit{bool}(F_\pi(A),R)(t_i)$, and $t_{0}(F_\pi(A)) = \mathbf{ff} = t_1(F_\pi(A))$.
The former and the inductive assumption give $x \in \mathit{bool}(A,R)(t_i)$.
\autoref{th:partialKleene}\eqref{Kleene:intro:item1} and the interpretations of $t_0$ and $t_1$ being monotone give $t_0(A) \sqsubseteq t_0( F_\pi A) = \mathbf{ff}$, and $t_1(A) \sqsubseteq t_1( F_\pi A) = \mathbf{ff}$, so $t_0(A) = \mathbf{ff} = t_1(A)$. Then $x \in \mathit{bool}(A,R)(t_0 \vee t_1)$.

\item \textsc{Inductive Step} ($t = t_0 \wedge t_1$). Fix $x \in \mathit{bool}( F_\pi(A),R)(t_0 \wedge t_1)$. This can happen according to one of the following cases.
\begin{itemize}
    \item There is $i \in\{0,1\}$ such that $x \in \mathit{bool}(F_\pi(A),R)(t_i)$, $\mathit{bool}( F_\pi(A),R)(t_{\neg i}) \neq \emptyset$, and $t_{0}(F_\pi(A)) = \mathbf{ff} = t_1(F_\pi(A))$.
    
    The inductive assumption on both $t_i$ and $t_{\neg i}$ gives us $\mathit{bool}(F_\pi(A),R)(t_i) \subseteq \mathit{bool}(A,R)(t_i)$ and $\mathit{bool}(F_\pi(A),R)(t_{\neg i}) \subseteq \mathit{bool}(A,R)(t_{\neg i})$, so we deduce that $x \in \mathit{bool}(A,R)(t_i)$ and $\mathit{bool}( A,R)(t_{\neg i}) \neq \emptyset$. Moreover, \autoref{th:partialKleene}\eqref{Kleene:intro:item1} and the interpretations of $t_0$ and $t_1$ being monotone give $t_0(A) \sqsubseteq t_0( F_\pi A) = \mathbf{ff}$, and $t_1(A) \sqsubseteq t_1( F_\pi A) = \mathbf{ff}$, so $t_0(A) = \mathbf{ff} = t_1(A)$. Then $x \in \mathit{bool}(A,R)(t_0 \wedge t_1)$.

    \item There is $i \in\{0,1\}$ such that $x \in \mathit{bool}(F_\pi(A),R)(t_i)$, $t_{i}(F_\pi(A)) = \mathbf{ff}$, and $t_{\neg i}(F_\pi(A)) = \mathbf{tt}$.
    
As above, the inductive assumption and \autoref{th:partialKleene}\eqref{Kleene:intro:item1} yield 
\begin{equation}\label{eq:bool:anti-monotone}
x \in \mathit{bool}(A,R)(t_i) \text{ and } t_{i}(A) = \mathbf{ff}.
\end{equation}
If additionally $t_{\neg i}(A) = \mathbf{tt}$, the latter and \eqref{eq:bool:anti-monotone} show that $x$ belongs to the first set in the union defining $\mathit{bool}(A,R)(t_0 \wedge t_1)$. Otherwise, $t_{\neg i}(A) = \mathbf{ff}$, so \eqref{eq:bool:anti-monotone} becomes
\begin{equation}\label{eq:bool:anti-monotone:2}
x \in \mathit{bool}(A,R)(t_i) \text{ and } t_{0}(A) = \mathbf{ff} = t_1(A).
\end{equation}
Our goal is to deduce that $x$ belongs to the second set in the union defining $\mathit{bool}(A,R)(t_0 \wedge t_1)$; to this end, it remains to show that $\mathit{bool}(A,R)(t_{\neg i}) \neq \emptyset$. 

Assume instead that $\mathit{bool}(A,R)(t_{\neg i}) = \emptyset$. As $\cc_A \subseteq R$ and $\mathit{bool}$ is sound by Proposition \ref{prop:Boolean:Ext:sound}, we deduce that $(z,t_{\neg i}) \notin \hc_A$ for all $z \in \V$. Then \autoref{improved:lemmaGR} implies that $t_{\neg i}(\mu F) = t_{\neg i} (A)$, so $t_{\neg i}(\mu F) = \mathbf{ff}$ by \eqref{eq:bool:anti-monotone:2}. But we were assuming that $t_{\neg i}(F_\pi(A)) = \mathbf{tt}$, so we obtain that $t_{\neg i}(\mu F) \sqsubset t_{\neg i}(F_\pi(A))$, that contradicts \autoref{th:partialKleene}\eqref{Kleene:intro:item1} and the interpretation of $t_{\neg i}$ being monotone. 
\end{itemize}
\end{trivlist}

Now fix $R \subseteq R' \subseteq \V \times \V$. We verify that $\mathit{bool}(A,R)(t) \subseteq \mathit{bool}(A,R')(t)$ by structural induction on $t \in \T$.
\begin{trivlist}
\item \textsc{Base Case} ($t = y \in \V$). 
Then by definition we have that $\mathit{bool}(A,R)(t) = \{ x \mid (x,y) \in R  \text{ and } A(y) = \mathbf{ff} \} \subseteq \{ x \mid (x,y) \in R'  \text{ and } A(y) = \mathbf{ff} \} = \mathit{bool}(A,R')(t)$.

\item \textsc{Base Case} (either $t = \mathbf{tt}$ or $t = \mathbf{ff}$).
Then the definition of $\mathit{bool}$ gives us $\mathit{bool}(A,R)(t) = \emptyset = \mathit{bool}(A,R')(t)$.

\item \textsc{Inductive Step} (either $t = t_0 \vee t_1$ or $t = t_0 \wedge t_1$). This step uses the inductive assumption that $\mathit{bool}(A,R)(t_i) \subseteq \mathit{bool}(A,R')(t_i)$, and so for example if $\mathit{bool}(A,R)(t_{\neg i}) \neq \emptyset$ then also $\mathit{bool}(A,R')(t_{\neg i}) \neq \emptyset$.%
\end{trivlist}
\end{proof}

\section{Missing proofs}
This appendix contains the proofs that were omitted from the main text.

\subsection{Missing proofs from Section~\ref{sec:prelim}.}
\begin{proof}[of Theorem~\ref{th:partialKleene}\eqref{Kleene:intro:item1}]
We have to prove that $F_\pi(\bot) \sqsubseteq F_x F_\pi(\bot)$ for all $\pi \in \V^*$ and $x \in V$, %
and we proceed by induction on $|\pi|$. If $|\pi| = 0$, then %
the inequality $\bot \sqsubseteq F_x(\bot)$ trivially holds.
For the induction step, fix $y \in \V$, %
for which our goal is proving that 
\begin{equation}\label{eq:double:ind:0}
F_y F_\pi (\bot) \sqsubseteq F_x F_y F_\pi (\bot).
\end{equation}

If $x = y$, then \eqref{eq:double:ind:0} becomes $F_x F_\pi (\bot) \sqsubseteq F_x F_x F_\pi (\bot)$. This follows from the inductive assumption $F_\pi (\bot) \sqsubseteq F_x F_\pi (\bot)$ and the fact that $F_x$ is monotone.

If $x \neq y$, then by definition of $F_x$, to prove \eqref{eq:double:ind:0} it is sufficient to compare the $x$-th coordinates of its left and right sides. Using the fact that $x \neq y$, we have $F_\pi (\bot) (x) = F_y F_\pi (\bot) (x)$, so \eqref{eq:double:ind:0} becomes 
\begin{equation}\label{eq:double:ind:1}
F_\pi (\bot) (x) \sqsubseteq F_x F_y F_\pi (\bot) (x).
\end{equation}

The inductive assumption gives $F_\pi(\bot) \sqsubseteq F_y F_\pi (\bot)$, and as $F_x$ is monotone we obtain $F_x F_\pi(\bot) \sqsubseteq F_x F_y F_\pi (\bot)$. Taking the $x$-th coordinates in the latter inequality, we obtain 
\begin{equation}\label{eq:double:ind:2}
F_x F_\pi(\bot)(x) \sqsubseteq F_x F_y F_\pi (\bot) (x ).
\end{equation}
Moreover, the inductive assumption also yields $F_\pi(\bot) \sqsubseteq F_x F_\pi (\bot)$, and so in particular
\begin{equation*}
F_\pi(\bot) (x) \sqsubseteq F_x F_\pi (\bot) (x).
\end{equation*}
The above and \eqref{eq:double:ind:2} imply \eqref{eq:double:ind:1}.
\end{proof}

\bigskip

Let us recall here that for an assignment $A\in \D^\V$ it holds that 
\begin{equation}\label{eq:fix-pointF:charact:viaFx}
A = F(A) \text{ if and only if } A = F_x (A) \text{ for every } x\in \V.
\end{equation}

\begin{proof}[of Theorem~\ref{th:partialKleene}\eqref{Kleene:intro:item2}]
We first show that $F_{\pi}(\bot) \sqsubseteq \mu F$ for every $\pi\in \V^*$, by induction on $|\pi|$.

If $|\pi|=0$, then $F_\pi$ is the identity, so $F_\pi (\bot) = \bot \sqsubseteq \mu F$.
For the inductive step, let $\pi = \pi' x \in \V^*$ for some $x \in \V$ and $\pi'\in \V^*$. 
The inductive assumption gives $F_{\pi'}(\bot) \sqsubseteq \mu F$, and the monotonicity of $F_x$ yields $F_x F_{\pi'}(\bot) \sqsubseteq F_x (\mu F)$. Now \eqref{eq:fix-pointF:charact:viaFx} implies $F_x (\mu F) = \mu F$, and so $F_x F_{\pi'}(\bot) \sqsubseteq \mu F$.

Fix $A \in \A$. For every $\pi\in \V^*$ we have $F_{\pi}(A) \in\A$, so for every $x \in \V$ Theorem~\ref{th:partialKleene}\eqref{Kleene:intro:item1} gives $F_{\pi}(A) \sqsubseteq F_{x} F_{\pi}(A)$.
As $\D^\V$ is Noetherian, there is $\pi\in \V^*$ such that $F_{\pi}(A) = F_x F_{\pi}(A)$ for every $x \in \V$. By \eqref{eq:fix-pointF:charact:viaFx}, we get that $F_{\pi}(A)$ is a fixed point of $F$, and so $\mu F\sqsubseteq F_{\pi}(A)$. By the first part of the proof, we conclude that $F_{\pi}(A) = \mu F$.
\end{proof}

\subsection{Missing proofs from Section~\ref{sec:Kleene:Iteration:with:Dependency Analysis}.}

\begin{proposition}\label{flows=unionOFfutureNows}
    For every $A \in \A$, we have that $\flows_A = \bigcup_{\pi \in \V^*} \cc_{F_\pi(A)}$.
\end{proposition}
\begin{proof}
This is an easy consequence of the definitions.
\end{proof}

\begin{proof}[of \autoref{flows:anti-monotonic}]
This is a corollary of \autoref{flows=unionOFfutureNows}.
\end{proof}

The last part of \autoref{flows:anti-monotonic} states that $\flows_{F_{x}(A)} \subseteq \flows_{A}$ for every $A \in \A$ and $x \in \V$. This is what we term the \emph{anti-monotonicity} property of the flow relation (strictly speaking, 
the function $\A \to \big(\wp(\V \times \V), {\subseteq} \big)$, mapping $A \mapsto \flows_{A}$, is weakly anti-monotone, see \autoref{def:weakly:anti-monotone}). Note that $\cc_A$ is not anti-monotone (\cf~\autoref{ex:now}, where $(x,x) \notin \cc_\bot$, but $(x,x) \in \cc_{F_z(\bot)}$).
In these terms, \autoref{flows=unionOFfutureNows} establishes that $\flows_A$ is the least relation containing $\cc_A$ that is anti-monotone.

\begin{proposition}\label{muF:via:flows:1}
Let $z\in \V$ and $A\in \A$ such that $A(z) = \mu F(z)$. Then,
\begin{enumerate}
\item \label{muF:via:flows:1:item1}
$(x,z) \notin \flows_A$ for every $x \in \V$;
\item \label{muF:via:flows:1:item2}
$(z,y) \notin \flows_A$ for every $y \in \V$.
\end{enumerate}
\end{proposition}
\begin{proof} 
Let us first prove that 
\begin{equation}\label{eq:Fpi(A)(z) = A(z)}
A(z) = F_\pi (A) (z) \text{ for every $\pi\in \V^*$.}
\end{equation} 
Indeed, as $\mu F$ is a fixed point of $F$, it is also a fixed point of $F_\pi$ for every $\pi\in \V^*$ by \eqref{eq:fix-pointF:charact:viaFx}, and so $F_\pi(A) \sqsubseteq F_\pi (\mu F) = \mu F$, by Theorem \ref{th:partialKleene}\eqref{Kleene:intro:item2}. Now Theorem \ref{th:partialKleene}\eqref{Kleene:intro:item1} gives $A \sqsubseteq F_\pi (A)$, and so $A \sqsubseteq F_\pi (A) \sqsubseteq \mu F$.
Since $A(z) = \mu F(z)$ by assumption, we deduce \eqref{eq:Fpi(A)(z) = A(z)}.
\begin{enumerate}
\item
Fix $x\in \V$, and arbitrary $\pi,\pi'\in \V^*$. Using \eqref{eq:Fpi(A)(z) = A(z)} twice we obtain
\[
F_{\pi'} F_\pi (A)(z) = A(z) = F_{\pi'} F_x F_\pi (A)(z),
\] hence $(x,z) \notin \flows_A$. %
\item
For an arbitrary $\pi\in \V^*$, applying \eqref{eq:Fpi(A)(z) = A(z)} twice again we get $F_\pi(A)(z) = A(z) = F_z F_\pi (A)(z)$. We also know that $F_\pi(A)(x) = F_z F_\pi (A)(x)$ for every other variable $x \in \V\setminus \{z\}$, and so $F_\pi(A) = F_z F_\pi (A)$. 
Then obviously $F_{\pi'} F_\pi (A) = F_{\pi'} F_z F_\pi(A)$ for every $\pi'\in \V^*$, showing that $(z,y) \notin \flows_A$.%
\end{enumerate}
\end{proof}

The following result essentially provides a converse to Proposition~\ref{muF:via:flows:1}.
\begin{proposition}\label{muF:via:flows:2}
Let $z\in \V$ and $A\in\A$ such that $A(z) \neq \mu F(z)$. Then,
\begin{enumerate}
\item \label{anticipated:lemmaGR:for:cc}
$(x,z) \in \cc_A$ for some $x \in \V$;
\item \label{(z,z):in:flows}
$(z,z) \in \flows_A$.
\end{enumerate}
\end{proposition}
\begin{proof}
1. Assume instead that $\{x\in \V \mid (x,z)\in \cc_A \} = \emptyset$. We first show by induction on $|\pi|$ that 
\begin{equation}\label{anticipated:eq:lemmaGR}
A(z) = F_\pi (A) (z) \text{ for every } \pi \in \V^*.
\end{equation}
If $|\pi|=0$, then $F_\pi$ is the identity, so $F_\pi (A) (z) = A(z)$ and there is nothing to prove. 
For the inductive step, let $\pi = x \pi' \in \V^*$ for some $x \in \V$ and $\pi'\in \V^*$. 

The inductive assumption gives $A(z) = F_{\pi'} (A) (z)$, and since $(x,z)\notin \cc_A$, in particular we have that $F_{\pi'}(A) (z) = F_{\pi'} F_x (A) (z) = F_{x\pi'}(A) (z)$, so \eqref{anticipated:eq:lemmaGR} holds for $\pi = x \pi'$.

By Theorem \ref{th:partialKleene}\eqref{Kleene:intro:item2}, there is $\pi \in \V^*$ such that $\mu F = F_\pi (A)$, and then \eqref{anticipated:eq:lemmaGR} implies that $A(z) = F_\pi (A) (z) = \mu F (z)$, a contradiction.

2. Obviously, the variable $x$ provided by the first item satisfies $(x,z) \in \cc_A \subseteq \flows_A$. Now we prove that also $(z,z) \in \flows_A$.

By Theorem~\ref{th:partialKleene}\eqref{Kleene:intro:item2}, there is $\pi_0 \in \V^*$ such that $\mu F = F_{\pi_0} (A)$, and we can assume that $\pi_0$ has minimal length such that the latter equality holds. Moreover, $\pi_0$ must contain $z$, otherwise $A(z) = F_{\pi_0} (A)(z) = \mu F (z)$, contradicting the assumption.
Write $\pi_0 = \pi z \pi'$, where the prefix $\pi$ of $\pi_0$ does not contain $z$. Then
\[
\mu F = F_{\pi z \pi'} (A) = F_{\pi'} F_z F_\pi(A),
\]
and $F_\pi(A) \neq F_z F_\pi(A)$ as $\pi_0$ has minimal length (otherwise, also $\mu F = F_{\pi \pi'} (A)$). Obviously $F_\pi(A) \neq F_z F_\pi(A)$ means that $F_\pi(A)(z) \neq F_z F_\pi(A) (z)$, so $\pi$ witnesses that $(z,z) \in \cc_{ F_\pi(A) } \subseteq \flows_A$. 
\end{proof}

Using Proposition~\ref{muF:via:flows:1} and Proposition~\ref{muF:via:flows:2}, we immediately obtain the following characterization of $\mu F$ via the relations of the form $\flows_A$ and $\cc_A$, thus proving~\autoref{th:flowsLocalFix}.

\begin{proof}[of \autoref{th:flowsLocalFix}]%
For $z\in \V$ and $A\in\A$, we have to prove that the following conditions are equivalent:
\begin{enumerate}
\item
$A(z) \neq \mu F(z)$; 
\item
$(x,z) \in \cc_A$ for some $x \in \V$;
\item
$(x,z) \in \flows_A$ for some $x \in \V$;
\item
$(z,z) \in \flows_A$;
\item 
$(z,y) \in \flows_A$ for some $y \in \V$.
\end{enumerate}
1. implies 2. by Proposition~\ref{muF:via:flows:2}(\ref{anticipated:lemmaGR:for:cc}), 2. implies 3. as $\cc_A \subseteq \flows_A$, 3. implies 1. by Proposition~\ref{muF:via:flows:1}(\ref{muF:via:flows:1:item1}). 
Moreover, 1. implies 4. by Proposition~\ref{muF:via:flows:2}(\ref{(z,z):in:flows}), 4. implies 5. trivially, and 5. implies 1. by Proposition~\ref{muF:via:flows:1}(\ref{muF:via:flows:1:item2}).
\end{proof}

The next result shows another striking difference between the flow relation and the now relation.
\begin{proposition}\label{char:xxINflows:xxINnow}
For $x \in \V$ and $A\in\A$, the following hold
\begin{enumerate}
\item \label{char:xxINflows:xxINnow:item1}
$(x,x) \in \flows_A$ if and only if  $A(x) \neq \mu F(x)$;
\item \label{char:xxINflows:xxINnow:item2}
$(x,x) \in \cc_A$ if and only if $A \neq F_x(A)$.
\end{enumerate}
\end{proposition}
\begin{proof}%

\eqref{char:xxINflows:xxINnow:item1}. This is the equivalence of conditions 1. and 4. from the above statement of \autoref{th:flowsLocalFix}.

\eqref{char:xxINflows:xxINnow:item2}.
Here we prove that the following conditions are equivalent:
\begin{enumerate}
\item
$A \neq F_x(A)$; %
\item
$(x,x) \in \cc_A$;
\item 
$(x,y) \in \cc_A$ for some $y \in \V$.
\end{enumerate}
If $A \neq F_x(A)$ then the empty path witnesses that $(x,x) \in \cc_A$, while 2. implies 3. is trivial.

Assume that $(x,y) \in \cc_A$ for some $y \in \V$, so there is $\pi \in \V^*$ such that $F_\pi(A)(y) \neq F_\pi F_x (A)(y)$. This means that $F_\pi(A) \neq F_\pi F_x (A)$, and so that also $A \neq F_x (A)$ as $F_\pi$ is a function.
\end{proof}

Regarding the oracles listed in \autoref{ex:EasyOracles}, $\mathit{ns}$ is the only one whose soundness is not self-evident. 
The oracle $\mathit{ns}$ is always available for every system of equations, and it depends solely on the assignment $A$.
As the name suggests, $\mathit{ns}$ states that only variables not currently ``stuck'' should be considered for local updates. 

\begin{proposition} \label{prop:not-stuckSound}
    The oracle $\mathit{ns}$ is sound.
\end{proposition}
\begin{proof}%
    The soundness of $\mathit{ns}$ follows by the observation that, by definition of $\cc_A$, if $(x,y) \in \cc_A$ then $A \neq F_x(A)$, \ie, $A(x) \neq f_x(A)$.
\end{proof}

\begin{proof}[of \autoref{thm:depKleene:correct}]
In the following, we use $v^{(i)}$ to denote the value of the variable $v$ at the $i$-th iteration of the while-loop (that is, $v^{(0)}$ is the value before entering the loop). For example, $A^{(0)} = \bot$, and by definition of $\A$ we have that $A^{(i)} \in \A$ for all $i \geq 0$ (see line~\ref{line:updateassign}). 

To prove termination, assume for a contradiction that $\id{todo}^{(i)} \neq \emptyset$, for all $i \geq 0$. We prove the following invariant, for all $i \geq 0$:
\begin{equation} \label{Main:text:eq:invariant}
       A^{(i)} \sqsubset A^{(i+1)}
       \lor 
       |\id{todo}^{(i+1)}| < |\id{todo}^{(i)}| %
       \,.
\end{equation}

Let $i \geq 0$. If the condition in line~\ref{li:increase} holds,  \ie, $A^{(i)}(x) \neq f_x(A^{(i)})$, then, by \autoref{th:partialKleene}\eqref{Kleene:intro:item1} and $A^{(i)} \in \A$, we conclude that $A^{(i)} \sqsubset A^{(i+1)}$ (see line~\ref{line:updateassign}). 
Otherwise, $A^{(i)} = A^{(i+1)}$ and $\id{todo}^{(i+1)} = \id{todo}^{(i)} \setminus \{ x \}$ for some $x \in todo^{(i)}$ extracted in line~\ref{li:extract}. Then $|\id{todo}^{(i+1)}| = |\id{todo}^{(i)}| - 1$.

Since $\D^\V$ is Noetherian, there exists $k$ such that $A^{(i+1)} =  A^{(i)}$ for all $i \geq k$. 
Since $\id{todo}^{(i)} \subseteq \V$ is finite for all $i$, now~\eqref{Main:text:eq:invariant} implies that $\id{todo}^{(h)} = \emptyset$ for some $h\geq 0$, contradicting the  assumption that $\id{todo}^{(i)} \neq \emptyset$ for all $i \geq 0$.

We prove that \proc{GlobalK} returns $\mu F(\tilde{x})$. Assume \proc{GlobalK} terminates after $h$ iterations (\ie, $\id{todo}^{(h)} = \emptyset$). We establish the following loop invariant by induction on $i \leq h$
\begin{equation}
\{y \mid (y,\tilde{x}) \in \cc_{A^{(i)}} \} \subseteq \id{todo}^{(i)} \, .
\label{eq:loopInv1}
\end{equation}
\begin{trivlist}
\item \textsc{Base Case} ($i = 0$). By construction, $\id{todo}^{(0)} \gets \{y \mid (y, \tilde{x}) \in \ora(A^{(0)}, \V \times \V) \}$ (see line~\ref{li:initTodo}). Since $\ora$ is sound we have $\{y \mid (y,\tilde{x}) \in \cc_{A^{(0)}} \} \subseteq \id{todo}^{(0)}$.
        
\item \textsc{Inductive Step}. Assume the inductive hypothesis holds for $i < h$. Let $x$ be the element extracted from $\id{todo}^{(i)}$ in line~\ref{li:extract}. 
We consider two cases. 
\begin{trivlist}
    \item If $A^{(i)}(x) = f_x(A^{(i)})$, then by definition $(x,\tilde{x}) \notin \cc_{A^{(i)}}$. Combining this with the inductive hypothesis, we have  
    $\{y \mid (y,\tilde{x}) \in \cc_{A^{(i)}} \} \subseteq \id{todo}^{(i)} \setminus \{x\}$. Additionally, $\id{todo}^{(i+1)} = \id{todo}^{(i)} \setminus \{x\}$ and $A^{(i+1)} = A^{(i)}$, so $\{y \mid (y,\tilde{x}) \in \cc_{A^{(i+1)}} \} \subseteq \id{todo}^{(i+1)}$.
            
    \item If $A^{(i)}(x) \neq f_x(A^{(i)})$, then $A^{(i+1)} = F_x(A^{(i)})$ and 
    $\id{todo}^{(i+1)} \gets \{y \mid (y, \tilde{x}) \in R \}$ for $R = \ora(A^{(i+1)}, \V \times \V)$ (see line~\ref{li:updatetodo}). Since $\ora$ is sound, $\{y \mid (y,\tilde{x}) \in \cc_{A^{(i+1)}} \} \subseteq \id{todo}^{(i+1)}$.
    \end{trivlist}        
\end{trivlist}
    
Finally, since $\id{todo}^{(h)} = \emptyset$, combining the loop invariant \eqref{eq:loopInv1} and \autoref{th:nowLocalFix}, we conclude that $A^{(h)}(\tilde{x}) = \mu F(\tilde{x})$.
\end{proof}

\begin{proof}[of \autoref{local:version:thm:depKleene:correct}]
In the following, we use $v^{(i)}$ to denote the value of the variable $v$ at the $i$-th iteration of the while-loop (that is, $v^{(0)}$ is the value before entering the loop). For example, $A^{(0)} = \bot$ and $\id{D}^{(0)} = \{ \tilde x\}$. By definition of $\A$, we have that $A^{(i)} \in \A$ for all $i \geq 0$ (see line~\ref{li:LocalKleene:AssUpdate}). 

To prove termination, assume for a contradiction that $\id{todo}^{(i)} \neq \emptyset$, for all $i \geq 0$.     
We prove the following invariant, for all $i \geq 0$:
\begin{equation} \label{eq:termination:localKleene}
    A^{(i)} \sqsubset A^{(i+1)}
       \lor
       \id{D}^{(i)} \subset \id{D}^{(i+1)}
       \lor %
       |\id{todo}^{(i+1)}| < |\id{todo}^{(i)}| %
       \,.
\end{equation}

Let $i \geq 0$. If line~\ref{li:LocalKleene:Increase} holds, \ie, either $A^{(i)}(x) \neq f_x(A^{(i)})$ or $\args(f_x) \not\subseteq\id{D}^{(i)}$, then we conclude respectively that either $A^{(i)} \sqsubset A^{(i+1)}$ (by \autoref{th:partialKleene}\eqref{Kleene:intro:item1}), or $\id{D}^{(i)} \subset\id{D}^{(i+1)}$ (see lines~\ref{li:LocalKleene:AssUpdate}--\ref{li:LocalKleene:DiscoveredUpdate}). 
Otherwise, $\id{todo}^{(i+1)} = \id{todo}^{(i)} \setminus \{ x \}$ for some $x \in todo^{(i)}$ extracted in line~\ref{li:LocalKleene:extract}, and so $|\id{todo}^{(i+1)}| = |\id{todo}^{(i)}| - 1$.

Since $\V$ is finite and $\D^\V$ is Noetherian, there is $k$ such that $A^{(i)} = A^{(i+1)}$ and $\id{D}^{(i)} = \id{D}^{(i+1)}$ for all $i \geq k$. 
Combined with \eqref{eq:termination:localKleene}, this implies that $\id{todo}^{(h)} = \emptyset$ for some $h\geq 0$, contradicting the  assumption that $\id{todo}^{(i)} \neq \emptyset$ for all $i \geq 0$.

To prove that $\proc{LocalK}(\E, \tilde{x}, \ora)$ returns $\mu F(\tilde{x})$ we %
first establish two more invariants:
\begin{enumerate}[label=(\alph*), ref=\alph*]
    \item \label{tildeXDiscovered} 
    $\tilde x \in \id{D}^{(i)}$ for all $i\geq 0$;
    \item \label{RisSound}
    $\flows_{A^{(i)}} \subseteq \id{R}^{(i)}$ for all $i\geq 0$.
\end{enumerate}
\eqref{tildeXDiscovered} is clear as $\tilde x \in \id{D}^{(0)}$, and $\id{D}^{(i)} \subseteq \id{D}^{(i+1)}$ for all $i\geq 0$ (see lines~\ref{li:LocalKleene:initialization}--\ref{li:LocalKleene:DiscoveredUpdate}). 
\\
\eqref{RisSound} 
Obviously $\flows_{A^{(0)}} \subseteq \V \times \V$, so $\flows_{A^{(0)}} \subseteq \ora(\id{V}^{(0)}, A^{(0)}, \V \times \V) = \id{R}^{(0)}$ by the soundness of $\ora$. Let $i \geq 0$, and assume the inductive assumption $\flows_{A^{(i)}} \subseteq \id{R}^{(i)}$. As $\flows_{A^{(i+1)}} \subseteq \flows_{A^{(i)}}$ by \autoref{flows:anti-monotonic}, we get $\flows_{A^{(i+1)}} \subseteq \id{R}^{(i)}$, and so the soundness of $\ora$ gives $\flows_{A^{(i+1)}} \subseteq \ora(\id{V}^{(i+1)}, A^{(i+1)}, \id{R}^{(i)})$. As %
$\id{R}^{(i+1)} = \ora(\id{V}^{(i+1)}, A^{(i+1)}, \id{R}^{(i)})$ by line~\ref{li:LocalKleene:RelationUpdate}, we have verified that $\flows_{A^{(i+1)}} \subseteq \id{R}^{(i+1)}$.

By the first part of the proof, assume $\proc{LocalK}$ terminates after $h$ iterations, \ie, $\id{todo}^{(h)} = \emptyset$. This occurs in one of the following two cases.

\textbf{Case 1.}
Either line~\ref{li:LocalKleene:initTodo} or line~\ref{li:LocalKleene:todoUpdate} sets
\[
\id{todo}^{(h)} = \big\{y \mid (y, \tilde{x}) \in \id{R}^{(h)} \big\} \cap \id{D}^{(h)}.
\]
From~\eqref{RisSound} %
we obtain that $\flows_{A^{(h)}} \subseteq \id{R}^{(h)}$, and so
\[
\{y \mid (y,\tilde{x}) \in \flows_{A^{(h)}} \} \cap \id{D}^{(h)} \subseteq \id{todo}^{(h)} = \emptyset.
\]
As $\tilde x \in \id{D}^{(h)}$ by~\eqref{tildeXDiscovered}, the above gives $(\tilde x, \tilde x) \notin \flows_{A^{(h)}}$.
By \autoref{char:xxINflows:xxINnow}\eqref{char:xxINflows:xxINnow:item1}, we conclude that $A^{(h)}(\tilde{x}) = \mu F(\tilde{x})$.

\textbf{Case 2.}
Otherwise, $\id{todo}^{(h)} = \emptyset$ for the repeated extractions of elements in line~\ref{li:LocalKleene:extract}, \ie, no element in $\id{todo}$ satisfies the condition in line \ref{li:LocalKleene:Increase}. Say there have been exactly $n$ consecutive extractions, with $0 < n \leq h$. %

Call $A = A^{(h-n)}$, $D = \id{D}^{(h-n)}$, $V = \id{V}^{(h-n)}$, $R = \id{R}^{(h-n)}$, and $\rest = \id{todo}^{(h-n)}$. The set $\rest$ had been set by either line~\ref{li:LocalKleene:initTodo} or line~\ref{li:LocalKleene:todoUpdate}, and so $\rest = \big\{y \mid (y, \tilde{x}) \in R \big\} \cap D$.
From~\eqref{RisSound} %
we obtain that $\flows_{A} \subseteq R$, and so
\begin{equation}\label{eq:soundness:todo:(h-n)}
\{y \mid (y,\tilde{x}) \in \flows_{A} \} \cap D \subseteq \rest.
\end{equation}

During each of these last $n$ iterations, the extraction of elements from $\rest$ does not change neither the set $D$, nor the assignment $A$, so for every $x \in \rest$ we have
\begin{equation}\label{rest:does not:change}
A(x) = f_x(A) \text{ and } \args( f_{x} ) \subseteq D.
\end{equation}
Fix $x \in \rest$, and let us verify that 
\[
\{y \mid (y,\tilde{x}) \in \flows_{A} \} \cap \args( f_{x} ) \subseteq \rest.
\]
Let $y \in \args( f_{x} )$ be such that $(y,\tilde{x}) \in \flows_{A}$. Then $y\in D$ by \eqref{rest:does not:change}, hence $y \in \rest$ by~\eqref{eq:soundness:todo:(h-n)}. 

By contradiction, assume that $A(\tilde{x}) \neq \mu F(\tilde{x})$. Then $(\tilde x, \tilde x) \in \flows_{A}$ by~\autoref{char:xxINflows:xxINnow}\eqref{char:xxINflows:xxINnow:item1}; moreover,~\eqref{tildeXDiscovered} gives $\tilde{x}\in D$, so~\eqref{eq:soundness:todo:(h-n)} implies that $\tilde{x}\in \rest$. This shows that $\rest$ is almost self-closed with respect to $\tilde{x}$ and $A$, so \autoref{thm:almost:self-closed} yields $A(\tilde{x}) = \mu F(\tilde{x})$, a contradiction.
\end{proof}

\begin{proof}[of \autoref{thm:almost:self-closed}]
Let $X \subseteq \V$ be an almost self-closed set with respect to $\tilde{x}$ and $A$.

\textbf{Step 1.} We first claim that for every $\pi \in \big( \{ y \mid (y,\tilde{x}) \in \flows_{ A } \} \big)^*$ we have 
\[
F_\pi(A)|_{ X \cup \{ y \mid (y, \tilde{x}) \notin \flows_A \} } = A|_{ X \cup \{ y \mid (y, \tilde{x}) \notin \flows_A \} }.
\]
We prove our claim by induction on the length of $\pi$. If $|\pi| = 0$, then $F_\pi(A) = A$ and there is nothing to prove.

Fix $\pi \in \big( \{ y \mid (y,\tilde{x}) \in \flows_{ A } \} \big)^*$ with $|\pi| > 0$, and $x \in X \cup \{ y \mid (y, \tilde{x}) \notin \flows_A \}$, for which we have to verify that $F_\pi(A)(x) = A(x)$. If $(x, \tilde{x}) \notin \flows_A$, the fact that $\pi \in \big( \{ y \mid (y,\tilde{x}) \in \flows_{ A } \} \big)^*$ guarantees that $x$ does not appear in $\pi$, and so $F_{\pi}(A)(x) = A(x)$. From now on assume $x \in X$. 

As $|\pi| > 0$, let $\pi = \pi' z$, for $\pi' \in \big( \{ y \mid (y,\tilde{x}) \in \flows_{ A } \} \big)^*$, and $(z,\tilde{x}) \in \flows_{ A }$. We have to check that $F_z F_{\pi'}(A)(x) = A(x)$. If $z \neq x$, then $F_z F_{\pi'}(A)(x) = F_{\pi'}(A)(x)$. As $x \in X$, the inductive assumption on $\pi'$ yields $F_{\pi'}(A)(x)= A(x)$, that proves $F_z F_{\pi'}(A)(x) = A(x)$. 

The final case to consider is when $z = x$, and then we have to prove that $F_x F_{\pi'}(A)(x) = A(x)$. By definition, this means checking that $f_x ( F_{\pi'}(A) ) = A(x)$. As $x \in X$ and $X$ is almost self-closed, it holds that $\args(f_x) \subseteq X \cup \{ y \mid (y, \tilde{x}) \notin \flows_A \}$, so we can apply the inductive assumption to $\pi'$ to obtain in particular that $F_{\pi'}(A)|_{ \args(f_x) } = A|_{ \args(f_x) }$. By definition of $\args(f_x)$, let $g$ be a function such that $f_x(B) = g(B|_{\args(f_x)})$ for every $B \in \A$. Then
\begin{align*}
f_x( F_{\pi'}(A) ) = g( F_{\pi'}(A)|_{\args(f_x)} ) = g( A|_{\args(f_x)} ) = f_x(A).
\end{align*}
Moreover, $f_x(A) = A$ as $x \in X$ and $X$ is almost self-closed. This concludes the proof of our claim. 

\textbf{Step 2.} We verify that $A(\tilde{x}) = \mu F(\tilde{x})$. %
As $A \in \A$, Theorem~\ref{th:partialKleene}\eqref{Kleene:intro:item2} ensures that there is $\pi \in \V^*$ such that $F_\pi ( A )(\tilde{x}) = \mu F(\tilde{x})$, 
and without loss of generality we can assume that $\pi$ has minimal length. Our goal is to prove that $\pi$ is empty. 

We first show that $\pi \in \big(\{ x \mid (x, \tilde{x}) \in \flows_{A} \} \big)^*$. %
If $\pi = \pi_1 x \pi_2$ for $\pi_1, \pi_2\in \V^*$ and $x \in \V$ such that $(x, \tilde{x}) \notin \flows_{A}$, then by definition we have
\[
F_{\pi_2} F_{\pi_1}( A )(\tilde{x}) = F_{\pi_2} F_x F_{\pi_1} ( A )(\tilde{x}) = F_{\pi} ( A )(\tilde{x}) = \mu F(\tilde{x}).
\]
The coincidence of the first and last members of the above equalities contradicts the fact that $\pi$ has minimal length. %
 
As $\tilde{x} \in X$ and $X$ is almost self-closed, it holds that $\args(f_{\tilde{x}}) \subseteq X \cup \{ y \mid (y, \tilde{x}) \notin \flows_A \}$. Then we can apply the first part of the proof to $\pi$, %
to obtain $F_\pi(A)|_{ \args(f_{\tilde{x}}) } = A|_{\args(f_{\tilde{x}})}$. By definition of $\args(f_{\tilde{x}})$, let $g$ be a function such that $f_{\tilde{x}}(B) = g(B|_{\args(f_{\tilde{x}})})$ for every $B \in \A$. Then
\[
f_{\tilde{x}}( F_\pi(A) ) = g(F_\pi(A)|_{\args(f_{\tilde{x}})}) = g(A|_{\args(f_{\tilde{x}})}) = f_{\tilde{x}}( A ).
\]
As $\tilde{x} \in X$ and $X$ is almost self-closed we have $f_{\tilde{x}}( A ) = A(\tilde{x})$. On the other hand, $F_\pi ( A )(\tilde{x}) = \mu F(\tilde{x})$ gives $f_{\tilde{x}}( F_\pi(A) ) = \mu F(\tilde{x})$. These yield $A(\tilde{x}) = \mu F(\tilde{x})$. %
\end{proof}

\subsection{Missing proofs from Section~\ref{sec:oracles}.}

\begin{proof}[of \autoref{prop:oracles:properties}]
The first statement is immediate.

When $\ora,\ora'$ are sound, it is very easy to verify that $\ora \cap \ora'$ is sound.
To show that $\ora \circ \ora'$ is sound, fix $R \subseteq \V \times \V$ such that $\cc_A \subseteq R$. From the soundness of $\ora'$ and $\ora$, we obtain in turn that $\cc_A \subseteq \ora'(A,R)$ and $\cc_A \subseteq \ora(A,\ora'(A,R))$.
\end{proof}

For the proof of \autoref{lem:soundnessMaxL/R}, we need a technical result first.
\begin{lemma}\label{remark:finally:used}
Let $\pi \in \V^*$ of minimal length such that $F_{\pi}(A)(y) \neq F_{\pi} F_x (A) (y)$.
Then $x\neq y$ if and only if $\pi$ is non-empty, and in this case, $\pi$ ends with $y$.
\end{lemma}
\begin{proof}
If $\pi$ is empty, then our assumption becomes $A(y) \neq F_x(A)(y)$, which implies $x=y$.

Now assume $x \neq y$. %
First, $\pi$ must contain $y$, otherwise 
$F_{\pi}(A)(y) = A(y) = F_{\pi} F_x (A) (y)$, that is excluded by our assumption. So, in particular $\pi$ is not empty. Let $\pi = \pi' y \pi''$, with $y$ not appearing in $\pi''$ (that is, we have isolated the last instance of $y$ in $\pi$, and so $F_{\pi''}$ does not change the $y$-th coordinate). Then the assumption gives 
\begin{align*}
F_y F_{\pi'} (A)(y) = F_{\pi''} F_y F_{\pi'} (A)(y) \neq F_{\pi''} F_y F_{\pi'} F_x(A)(y) = F_y F_{\pi'} F_x(A)(y).
\end{align*}
The first and the last members above show that $F_y F_{\pi'} (A)(y) \neq F_y F_{\pi'} F_x(A)(y)$. By minimality of $\pi$ we must have $\pi = \pi'y$.
\end{proof}

As an application of \autoref{remark:finally:used}, the following result more directly characterizes the membership $(x,y) \in \cc_A$ in terms of the right-hand side of the equation $y = f_y$ associated with the variable that gets influenced.
\begin{proposition}\label{prop:pairsINflows}
For every $A\in\A$, we have $(x,y) \in \cc_A$ if and only if either one of the following holds:
\begin{enumerate}
\item \label{prop:pairsINflows:=}
$x = y$ and $A(y) \neq f_y(A)$;
\item \label{prop:pairsINflows:neq}
$x \neq y$ and $f_y ( F_{\pi}(A) ) \neq f_y ( F_{\pi} F_x(A) )$ for some $\pi \in \V^*$.
\end{enumerate}
\end{proposition}
\begin{proof}
Let $\pi \in \V^*$ be such that $F_{\pi}(A)(y) \neq F_{\pi} F_x (A) (y)$. Without loss of generality, assume $\pi$ to be of minimal length satisfying the latter inequality.

If $x = y$, then Lemma \ref{remark:finally:used} yields that $\pi$ is empty, and so 
$A(y) \neq F_y(A)(y)$. As $F_y(A)(y) = f_y(A)$, we conclude this case.

If $x \neq y$, then Lemma \ref{remark:finally:used} yields that $\pi$ has the form $\pi = \pi'y$, and so $F_y F_{\pi'} (A)(y) \neq F_y F_{\pi'} F_x(A)(y)$. The latter means exactly that $f_y ( F_{\pi'}(A) ) \neq f_y ( F_{\pi'} F_x(A) )$.
\end{proof}

As their names suggest, $\mathit{max}$ is more accurate than $\mathit{smax}$.
\begin{lemma}\label{max:less:than:smax}
    $\mathit{max} \preceq \mathit{smax}$.
\end{lemma}
\begin{proof}%
Let $A \in \A$ and $R \subseteq \V \times \V$. To show that $\mathit{max}(A, R) \subseteq \mathit{smax}(A, R)$, fix $(x,y) \in \mathit{max}(A,R)$. We distinguish two cases:
\begin{description}%
    \item[Case 1]: $y=x$ and $A(x) \neq f_x(A)$. To conclude that $(x,x) \in \mathit{smax}(A,R)$, we have to verify that $A(x)$ is not maximal. By contradiction, assume that 
    $A(x)$ is maximal. Then, by \autoref{th:partialKleene}\eqref{Kleene:intro:item1} 
    we have $A(x) \sqsubseteq F_x(A)(x)$. As $A(x)$ is maximal and $F_x(A)(x) = f_x(A)$, we get $A(x) = f_x(A)(x)$, a contradiction.

    \item[Case 2]: $A(x)$ and $f_y(A)$ are not maximal. To conclude
    $(x,y) \in \mathit{smax}(A,R)$, we verify that $A(y)$ is not maximal. 
    For the sake of finding a contradiction, assume $A(y)$ is maximal. Then, similarly to the previous case, we can deduce that $A(y) = F_y(A)(y)$. 
    By definition of $F$, we have $F_y(A)(y) = f_y(A)$. Thus also $f_y(A)$ is maximal, a contradiction.%
\end{description}
\end{proof}

\begin{proof}[of \autoref{lem:soundnessMaxL/R}]
The soundness of $\mathit{max}_r$ follows by showing that for every $A \in \A$ the following inclusions hold:
\begin{enumerate}[label=\emph{\alph*}., ref=\emph{\alph*}]
    \item $\{ (x,y) \in \cc_A \mid x = y \} \subseteq \{ (x,x) \mid A(x) \neq f_x(A) \}$; \label{1_MaxL/R}
    \item $\{ (x,y) \in \cc_A \mid x \neq y \} \subseteq \{ (x,y) \mid x \neq y \text{ and $f_y(A)$ not maximal} \}$. \label{3_MaxL/R}
\end{enumerate}

The inclusion \eqref{1_MaxL/R} follows directly from \autoref{prop:pairsINflows}\eqref{prop:pairsINflows:=}. 
For the inclusion \eqref{3_MaxL/R}, let $(x,y) \in \cc_A$ with $x \neq y$.
By \autoref{prop:pairsINflows}\eqref{prop:pairsINflows:neq} 
there exists a sequence $\pi \in \V^*$ such that $f_y(F_\pi(A)) \neq f_y(F_\pi F_x(A))$. By the monotonicity of $f_y$ and \autoref{th:partialKleene}\eqref{Kleene:intro:item1}, we have $f_y(A) \sqsubseteq f_y(F_\pi(A)) \sqsubset f_y(F_\pi F_x(A))$, so $f_y(A)$ is not maximal.

In view of \eqref{1_MaxL/R}, to establish the soundness of $\mathit{max}_l$ it suffices to show that for every $A \in \A$ the following inclusion holds:
\begin{enumerate}%
    \item[] $\{ (x,y) \in \cc_A \mid x \neq y \} \subseteq \{ (x,y) \mid x \neq y \text{ and $A(x)$ not maximal} \}$. \label{nuova3_MaxL/R}
\end{enumerate}
Assume $(x,y) \in \cc_A$ with $x \neq y$. By definition of $\cc_A$ we have $A \neq F_x(A)$. Then, by \autoref{th:partialKleene}\eqref{Kleene:intro:item1}, it follows $A(x) \sqsubset F_x(A)(x)$, which implies that $A(x)$ is not maximal.

Let $\mathit{max}' = \mathit{max}_l \cap \mathit{max}_r$. Then $\mathit{max}'$ is sound by \autoref{prop:oracles:properties}\eqref{operation-closure} and the first part of the proof. It is immediate to see that $\mathit{max}' \preceq \mathit{max}$, while $\mathit{max} \preceq \mathit{smax}$ by Lemma \ref{max:less:than:smax}. So $\mathit{max}$ and $\mathit{smax}$ are sound by Proposition \ref{prop:oracles:properties}\eqref{upclosure}.  
\end{proof}

\begin{example}\label{ex:max'<max}
    We remark that the oracle $\mathit{max}_l \cap \mathit{max}_r$
    is strictly more accurate than $\mathit{max}$.
    Indeed, consider the systems of equations from \autoref{ex:locality}.
    For the assignment $F_z(\bot)$ and an arbitrary $R \subseteq \V \times \V$ we have:
    \begin{equation*}
        (\mathit{max}_l \cap \mathit{max}_r)(F_z(\bot), R)
        = \{(x,x), (x,y), (w,y), (w,w) \} 
        \subsetneq 
        \mathit{max}(F_z(\bot), R) \,.
    \end{equation*}
\end{example}

\begin{proof}[of \autoref{improved:lemmaGR}]
We prove the two implications separately.

($\Leftarrow$) 
We first show by induction on $|\pi|$ that our assumption yields
\begin{equation}\label{eq:improved:lemmaGR}
t(A) = t ( F_\pi (A) ) \text{ for every } \pi \in \V^*.
\end{equation}
\begin{trivlist}
\item[\textsc{Base Case.}] $\pi=\varepsilon$, then $F_\pi$ is the identity and $t( F_\pi (A) ) = t (A)$.
\item[\textsc{Inductive Step.}] Let $\pi = y \pi' \in \V^*$. By inductive hypothesis $t(A) = t ( F_{\pi'}(A) )$, and since $(y,t)\notin \hc_A$, we also conclude $t ( F_{\pi'}(A) ) = t ( F_{\pi'} F_y (A) ) = t ( F_{y\pi'} (A) )$. Thus,~\eqref{eq:improved:lemmaGR} holds for $\pi = y \pi'$.
By~\autoref{th:partialKleene}\eqref{Kleene:intro:item2}, there is $\pi \in \V^*$ such that $\mu F = F_\pi (A)$, and then~\eqref{eq:improved:lemmaGR} implies that $t(A) = t ( F_\pi (A) ) = t ( \mu F)$.

($\Rightarrow$)
Let $A \in \A$ such that $t(A) = t(\mu F)$. 
By~\autoref{th:partialKleene}, for all $\pi \in \V^*$ we have $A \sqsubseteq F_\pi(A) \sqsubseteq \mu F$. Since $t$ has monotonic interpretation, 
 $   t(A) \sqsubseteq t(F_\pi(A)) \sqsubseteq t(\mu A)$.
In combination with $t(A) = t(\mu F)$, we get that $t(A) = t(F_\pi(A))$, for all $\pi\in\V^*$.
Consequently, $t(F_\pi(A)) = t(A) = t(F_{x\pi}(A)) = t(F_{\pi}F_x(A))$,
for all $x \in \V$ and $\pi\in\V^*$.
By definition of $\hc_A$ we conclude that $(x,t)\notin\hc_A$ for every $x\in\V$.%
\end{trivlist}
\end{proof}

\begin{proof}[of \autoref{prop:sound:oracles:from:liftings}]
Fix $A \in \A$ and $R \subseteq \V \times \V$ such that $\cc_A \subseteq R$. We verify that $\cc_A \subseteq \ora[\lift](A, R)$. To this end consider a pair $(x,y) \in \cc_A$.

If $x = y$, then $A \neq F_x(A)$ by~\autoref{char:xxINflows:xxINnow}\eqref{char:xxINflows:xxINnow:item2}, and so $(x,x) \in \ora[\lift](A, R)$.

If $x \neq y$, then from $(x,y) \in \cc_A$ and \autoref{prop:pairsINflows}\eqref{prop:pairsINflows:neq} there is $\pi \in \V^*$ such that  $t_y ( F_{\pi} (A) ) \neq t_y ( F_{\pi} F_x(A) )$, and so $(x,t_y) \in \hc_A$. By soundness of $\lift$, we have $x \in \lift(A, R)(t_y)$. Then $(x,y) \in \ora[\lift](A,R)$ by definition.
\end{proof}

The next result formalizes the idea that when a variable $x$ has an influence on a term $t$, it is because $x$ has an influence on some subterm of $t$ (and so, on some variable appearing in $t$).
\begin{lemma}\label{property:of:hc}
Let $A \in \A$. If $(x,t) \in \hc_A$, then $(x,y) \in \cc_A$ for some $y \in \var(t)$.
\end{lemma}
\begin{proof}
    We proceed by induction on the structure of the term $t$. 
    \begin{trivlist}
    \item \textsc{Base Case} ($t = y \in \V$). %
    It immediately follows from the fact that $\hc_A$ coincides with $\cc_A$ on the variables.
    \item \textsc{Inductive Step} ($t = \sigma(t_1,\twodots,t_n)$). 
    Assume $(x, t) \in \hc_A$ then the following hold:
    \begin{align*}
        &(x, t) \in \hc_A \iff \exists \pi'.\, t(F_{\pi'}(A)) \neq t(F_{x\pi'}(A)) \tag{def. $\hc_A$}\\
        &\iff \exists \pi'.\, \sigma_I(t_1(F_{\pi'}(A)),\twodots, t_n(F_{\pi'}(A))) \neq \sigma_I(t_1(F_{x\pi'}(A)),\twodots, t_n(F_{x\pi'}(A))) \\
        &\implies \exists i \in \{1\twodots n \} .\, \exists \pi'.\, t_i(F_{\pi'}(A)) \neq t_i(F_{x\pi'}(A)) \tag{$\sigma_I$ is a function} \\
        &\iff \exists i \in \{1\twodots n \} .\, (x, t_i) \in \hc_A \tag{def. $\hc_A$} \\
        &\implies \exists i \in \{1\twodots n \} .\, \exists y \in \var(t_i) .\, (x, y) \in \cc_A \tag{ind. hp} \\
        &\iff \exists y \in \var(t) .\, (x, y) \in \cc_A. \tag{as $\var(t) = \bigcup_{i = 1}^n \var(t_i)$}
    \end{align*}
    \end{trivlist}  
\end{proof}

\begin{proof}[of \autoref{lem:extension:closure:prop}]
Let $A \in \A$ and $R \subseteq \V \times \V$ be such that $\cc_A \subseteq R$. 

We first verify that $\lift \cap \lift'$ is sound. And indeed 
\begin{align*}
\hc_A 
&\subseteq \{ (x,t) \mid x \in \lift(A, R)(t) \}
        \cap \{ (x,t) \mid x \in \lift(A, R)(t) \}
\tag{$\lift,\lift'$ sound} \\
& = \{ (x,t) \mid x \in (\lift\cap\lift')(A, R)(t) \}\,,
\end{align*}
which corresponds to the statement of soundness for $\lift\cap\lift'$.

To check that $\argUnion{\lift}$ is sound, we fix $(x,t)\in\hc_A$, and we verify that $x \in \argUnion{\lift}(A, R)(t)$. By \autoref{property:of:hc}, there is $y \in \var(t)$ such that $(x,y) \in \cc_A$, and so $(x,y) \in \hc_A$. From the assumption that $\cc_A \subseteq R$ and that $\lift$ is sound we obtain that $x \in \lift(A, R)(y)$, and so $x \in \argUnion{\lift}(A, R)(t)$ by definition.
\end{proof}

\begin{proof}[of \autoref{prop:Boolean:Ext:sound}]
Let $A \in \A$ and $R \subseteq \V \times \V$ be such that $\cc_A \subseteq R$. To prove $\mathit{bool}$ is a sound extension we show $\hc_A \subseteq \{ (x,t) \mid x \in \mathit{bool}(A,R)(t) \}$. Writing down the definition of $\hc_A$, this entails checking that %
    \begin{equation}\label{bool:oracle:fundamental:prop}
        \forall t \in \T. \forall x \in \V. \,
        \big( \exists \pi \in \V^* .
        t(F_\pi (A)) \neq t(F_\pi F_x(A)) 
        \Rightarrow x \in \mathit{bool}(A,R)(t) \big).
    \end{equation}
We prove \eqref{bool:oracle:fundamental:prop} by structural induction on $t \in \T$.
\begin{trivlist}
\item \textsc{Base Case} ($t = y \in \V$). Fix $x \in \V$ and $\pi \in \V^*$ such that $t(F_\pi (A)) \neq t(F_\pi F_x(A))$, \ie, $F_\pi (A)(y) \neq F_\pi F_x(A)(y)$. This means that $(x,y) \in \cc_A \subseteq R$ and that $F_\pi (A)(y) = \mathbf{ff}$. As $A(y) \sqsubseteq F_\pi (A)(y)$ by \autoref{th:partialKleene}\eqref{Kleene:intro:item1}, we obtain that $A(y) = \mathbf{ff}$, so by definition $x \in \mathit{bool}(A,R)(y)$.

\item \textsc{Base Case} (either $t = \mathbf{tt}$ or $t = \mathbf{ff}$).
The implication \eqref{bool:oracle:fundamental:prop} holds trivially because $t$ is interpreted as a constant map.

\item \textsc{Inductive Step} ($t = t_0 \vee t_1$). Fix $x \in \V$ and $\pi \in \V^*$ such that
    \begin{equation*}
        (t_0 \vee t_1)(F_\pi (A)) \neq (t_0 \vee t_1)(F_\pi F_x(A)) \,.
    \end{equation*}
Then we obtain that $(t_0 \vee t_1)(F_\pi (A)) = \mathbf{ff}$ and $(t_0 \vee t_1)(F_\pi F_x(A)) = \mathbf{tt}$, so
\begin{gather}
t_0 (F_\pi(A)) \vee t_1(F_\pi(A)) = \mathbf{ff},		\label{bool:eq3}
\\
t_0 (F_\pi F_x(A)) \vee t_1(F_\pi F_x(A)) = \mathbf{tt}.		\label{bool:eq4}
\end{gather}
Now \eqref{bool:eq3} implies that $t_0 (F_\pi(A)) = t_1(F_\pi(A)) = \mathbf{ff}$. For $j \in \{0,1\}$, \autoref{th:partialKleene}\eqref{Kleene:intro:item1} and the interpretation of $t_j$ being monotone give $t_j(A) \sqsubseteq t_j( F_\pi(A) ) = \mathbf{ff}$, and so %
\begin{equation}\label{eq:bool:final}
t_0(A) = \mathbf{ff} = t_1(A).
\end{equation}

On the other hand, by \eqref{bool:eq4} there is $i \in \{0,1\}$ such that $t_i(F_\pi F_x(A)) = \mathbf{tt}$. So $t_i (F_\pi(A)) = \mathbf{ff} \neq \mathbf{tt} = t_i(F_\pi F_x(A))$ shows that $(x,t_i) \in \hc_A$, and the inductive assumption gives that 
\begin{equation*}
x \in \mathit{bool}(A,R)(t_i).
\end{equation*} 
The latter and \eqref{eq:bool:final} yield that $x \in \mathit{bool}(A,R)(t_0 \vee t_1)$ by definition.

\item \textsc{Inductive Step} ($t = t_0 \wedge t_1$). Fix $x \in \V$ and $\pi \in \V^*$ such that
\begin{gather}
t(F_\pi (A)) \neq t(F_\pi F_x(A)). \label{eq:BoolOracle:new}
\end{gather}
The above means that $t_0 (F_\pi (A))\wedge t_1(F_\pi (A)) \neq t_0(F_\pi F_x(A)) \wedge t_1(F_\pi F_x(A))$, so there is $i \in \{0,1\}$ such that $t_i(F_\pi(A)) \neq  t_i(F_\pi F_x(A))$, and then the inductive assumption gives 
\begin{equation}\label{eq:bool:0}
x \in \mathit{bool}(A,R)(t_i).
\end{equation}
Moreover, $t_i(F_\pi(A)) \neq  t_i(F_\pi F_x(A))$ also implies that $t_i(F_\pi(A)) = \mathbf{ff}$. \autoref{th:partialKleene}\eqref{Kleene:intro:item1} and the interpretation of $t_i$ being monotone give $t_i(A) \sqsubseteq t_i( F_\pi(A) ) = \mathbf{ff}$, so
\begin{equation}\label{eq:bool:feb:03}
t_i(A) = \mathbf{ff}.
\end{equation}

Let $j \neq i$ in $\{0,1\}$. If $t_j(A) = \mathbf{tt}$, then the latter, \eqref{eq:bool:0} and \eqref{eq:bool:feb:03} give that $x \in \mathit{bool}(A,R)(t_0 \wedge t_1)$ by definition and we conclude. Otherwise, $t_j(A) = \mathbf{ff}$. Again in view of \eqref{eq:bool:0} and \eqref{eq:bool:feb:03}, to deduce that $x \in \mathit{bool}(A,R)(t_0 \wedge t_1)$ by definition we only need to check that $\mathit{bool}(A,R)(t_j) \neq \emptyset$. 

By contradiction, assume instead that $\mathit{bool}(A,R)(t_j) = \emptyset$. Then the inductive assumption gives that $(z,t_j)\notin \hc_A$ for every $z\in\V$, so Theorem \ref{improved:lemmaGR} yields that $t_j( \mu F) = t_j(A) = \mathbf{ff}$. As $t = t_0 \wedge t_1$, also $t( \mu F ) = \mathbf{ff}$. 
\autoref{th:partialKleene}\eqref{Kleene:intro:item1} and the interpretation of $t$ being monotone give $t( F_\pi A) \sqsubseteq t( \mu F ) = \mathbf{ff}$, and so $t( F_\pi A) = \mathbf{ff}$. Similarly, also $t( F_\pi F_x A) \sqsubseteq t( \mu F ) = \mathbf{ff}$, and so $t( F_\pi F_x A) = \mathbf{ff} = t( F_\pi A)$, contradicting \eqref{eq:BoolOracle:new}.%
\end{trivlist}
\end{proof}

\begin{theorem}\label{improved:lemmaGR-flow}
For $t \in \T$ and $A\in\A$, %
$
t(A) = t(\mu F) \iff 
\forall x \in \V.\, (x,t)\notin \hflows_A
$.
\end{theorem}
\begin{proof}
Similarly to~\autoref{flows=unionOFfutureNows}, for every $A\in \A$ we have ${\hflows_A} = \bigcup_{\pi \in \V^*} \hc_{ F_{\pi}(A) }$.

($\Leftarrow$) If $(x,t)\notin\hflows_A$ for every $x\in\V$, then $(x,t)\notin\hc_A$ for every $x\in\V$, and so $t(A) = t ( \mu F)$ by \autoref{improved:lemmaGR}.

($\Rightarrow$) Assume that $t(A) = t ( \mu F)$. Fix $\pi \in \V^*$. Since $t$ has monotonic interpretation, we have that $t(A) \sqsubseteq t(F_\pi(A)) \sqsubseteq t(\mu F)$, and so also $t(F_\pi(A)) = t(\mu F)$. Now \autoref{improved:lemmaGR} gives $(x,t)\notin\hc_{F_\pi(A)}$ for every $x\in\V$. As the latter holds for an arbitrary $\pi \in \V^*$, we deduce that $(x,t)\notin \hflows_A$ for every $x \in \V$.
\end{proof}

\begin{proof}[of \autoref{LOCAL:prop:sound:oracles:from:liftings}]
Fix $A \in \A$ and $R \subseteq \V \times \V$ such that $\flows_A \subseteq R$. We have to verify that $\flows_A \subseteq \ora[\lift](V,A, R)$, and to this end consider a pair $(x,y) \in \flows_A$.
If $y \notin V$, then $(x,y) \in \ora[\lift](V,A, R)$ by definition, so from now on we let $y \in V$. 

Assume that $x \neq y$ and $y \in V$. From $(x,y) \in \flows_A$, let $\pi, \pi'\in \V^*$ be such that $F_{\pi'} F_{\pi}(A)(y) \neq F_{\pi'} F_x F_{\pi}(A) (y)$, and without loss of generality assume that $\pi'$ is minimal satisfying the latter inequality.
Then \autoref{remark:finally:used} yields that $\pi'$ has the form $\pi' = \pi''y$, and so 
$F_y F_{\pi''} F_{\pi}(A)(y) \neq F_y F_{\pi''} F_x F_{\pi}(A) (y)$.
The latter means exactly that $t_y ( F_{\pi''} F_{\pi}(A) ) \neq t_y ( F_{\pi''} F_x F_{\pi}(A) )$, and so $(x,t_y) \in \hflows_A$. 
By local soundness of $\lift$, we have $x \in \lift(A, R)(t_y)$. Then $(x,y) \in \ora[\lift](V,A,R)$ by definition.

Finally, consider the case when $x = y \in V$. If additionally $(z, t_x) \in \lift(A,R)$ for some $z$, then $(x,x) \in \ora[\lift](V,A, R)$ by definition. 
Otherwise, $x = y \in V$, and $(z, t_x) \notin \lift(A,R)$ for every $z$. As $\lift$ is locally sound, we obtain that $(z, t_x) \notin \hflows_A$ for every $z$, and then \autoref{improved:lemmaGR-flow} gives $t_x(A) = t_x(\mu F)$. 
On the other hand, the initial assumption $(x,x) \in \flows_A$ yields $A(x) \neq \mu F(x)$ by \autoref{char:xxINflows:xxINnow}, and so
\[
A(x) \neq \mu F(x) = t_x (\mu F) = t_x (A).
\]
This shows that $A(x) \neq t_x (A)$, and so $(x,x) \in \ora[\lift](V,A, R)$ by definition.
\end{proof}

\begin{proof}[of \autoref{extension:sound+anti-mon:inAssign:implies:locally:sound}]
Fix $A \in \A$ and $R \subseteq \V \times \V$ such that $\flows_A \subseteq R$. We have to verify that $\hflows_A \subseteq \{ (x,t) \mid x \in \lift(A, R)(t) \}$, and to this end consider a pair $(x,t) \in \hflows_A$. By definition, there is $\pi \in \V^*$ such that $(x,t) \in \hc_{ F_{\pi}(A) }$. As $\cc_{ F_{\pi}(A) } \subseteq \flows_A \subseteq R$ and $\lift$ is sound, we obtain that $x \in \lift(F_{\pi}(A),R)(t)$. Finally, the assumption on $\lift$ yields %
$\lift(F_\pi(A),R)(t) \subseteq \lift(A,R)(t)$, hence $x \in \lift(A,R)(t)$. 
\end{proof}

\end{document}